%% file: JSAC-manuscript.tex
\documentclass[lettersize,journal]{IEEEtran}
\usepackage{amsmath,amsfonts}
\usepackage{array}
\usepackage[caption=false,font=normalsize,labelfont=sf,textfont=sf]{subfig}
\usepackage{textcomp}
\usepackage{stfloats}
\usepackage{url}
\usepackage{verbatim}
\usepackage{graphicx}
\usepackage{cite}
\usepackage{pifont}
\usepackage{multirow}
\usepackage{threeparttable}
\usepackage{verbatim}
\usepackage{booktabs}
\usepackage{array}
\usepackage{amsmath,amssymb,amsfonts}
\usepackage[ruled,noend,linesnumbered]{algorithm2e}

\usepackage{algpseudocode}
\usepackage{graphicx}
\usepackage{xcolor}
\usepackage{graphicx}
\usepackage{textcomp}
\usepackage{xcolor}
\usepackage{enumitem}
\usepackage{url}
\usepackage{bm}
\usepackage{makecell}
\definecolor{red}{rgb}{1,0,0} 
\definecolor{blue}{rgb}{0,0,0}
\newtheorem{theorem}{Theorem}
\newtheorem{lemma}{Lemma}
\newtheorem{definition}{Definition} 
\newtheorem{proof}{Proof}
\newcommand{\vs}[1]{{\color{blue} #1}}

\begin{document}

\title{\textit{EdgeMatrix}: A Resource-Redefined Scheduling Framework for SLA-Guaranteed Multi-Tier Edge-Cloud Computing Systems}

\author{Shihao~Shen,
		Yuanming~Ren,
		Yanli~Ju,
		Xiaofei~Wang,~\IEEEmembership{Senior Member,~IEEE,}
		Wenyu~Wang,
		and~Victor~C.M.~Leung,~\IEEEmembership{Life~Fellow,~IEEE}%
\thanks{%
	 Manuscript received May 16, 2022; revised September 3, 2022; accepted October 25, 2022.
	 This work was supported in part by the National Key Research and Development Program of China under Grant No.2019YFB2101901, in part by the National Science Foundation of China under Grant No.62072332, and in part by the Tianjin Xinchuang Haihe Lab under Grant No.22HHXCJC00002.
	 This paper was presented in part at the IEEE INFOCOM, Virtual Conference, 2022\cite{edgematrix}. \textit{(Corresponding author: Xiaofei Wang.)}

Shihao~Shen, Yuanming~Ren, Yanli~Ju and Xiaofei~Wang are with the College of Intelligence and Computing, Tianjin University, Tianjin 300350, China (email: {shenshihao, renyuanming, yanliju, xiaofeiwang}@tju.edu.cn).

Wenyu~Wang is with the Shanghai Zhuichu Networking Tecnologies Co., Ltd., Shanghai 200120, China (email: wayne@pplabs.org).

Victor C.M. Leung is with the College of Computer Science and Software Engineering, Shenzhen University, Shenzhen 518052, China (e-mail: vleung@ieee.org).

}%

}

\markboth{IEEE JOURNAL ON SELECTED AREAS IN COMMUNICATIONS,~Vol.~xx, No.~x, xxxx~2023}{Shen \MakeLowercase{\textit{et al.}}: \textit{EdgeMatrix}:  A Resource-Redefined Scheduling Framework for SLA-Guaranteed Multi-Tier Edge-Cloud Computing Systems}

\maketitle

\begin{abstract}

\vs{With the development of networking technology, the computing system has evolved towards the multi-tier paradigm gradually. However, challenges, such as multi-resource heterogeneity of devices, resource competition of services, and networked system dynamics, make it difficult to guarantee service-level agreement (SLA) for the applications. In this paper, we propose a multi-tier edge-cloud computing framework, \textit{EdgeMatrix}, to maximize the throughput of the system while guaranteeing different SLA priorities. First, in order to reduce the impact of physical resource heterogeneity, \textit{EdgeMatrix} introduces the Networked Multi-agent Actor-Critic (NMAC) algorithm to re-define physical resources with the same quality of service as logically isolated resource units and combinations, i.e., cells and channels. In addition, a multi-task mechanism is designed in \textit{EdgeMatrix} to solve the problem of Joint Service Orchestration and Request Dispatch (JSORD) for matching the requests and services, which can significantly reduce the optimization runtime. For integrating above two algorithms, \textit{EdgeMatrix} is designed with two time-scales, i.e., coordinating services and resources at the larger time-scale, and dispatching requests at the smaller time-scale. Realistic trace-based experiments proves that the overall throughput of \textit{EdgeMatrix} is 36.7\% better than that of the closest baseline, while the SLA priorities are guaranteed still.}
\end{abstract}


\input{jsac_body}

\bibliographystyle{IEEEtran}

\bibliography{JSAC2022} 
 
\vspace{11pt}

\vspace{-3pt}

\vfill

\end{document}

%% file: jsac_body.tex
\section{Introduction}\label{sec:Introduction}

\subsection{Background and Problem Statement}
\label{subsec:Background and Problem Statement}

\IEEEPARstart{W}{ith} the explosive growth of networked devices and services such as autonomous driving, virtual reality and smart city, the traditional network architecture is facing challenges~\cite{nanakkal2021brief, chakareski2020multi, wu2020collaborate}. According to the GSMA's \textit{The Mobile Economy 2020} report, the number of global iot connections will increase from 12 billion in 2019 to nearly 25 billion by 2025\cite{GSMA2020}. Therefore, the centralized cloud computing faces the following challenges: (\textit{$\romannumeral1$}) the rapid growth of network devices leads to more computing requests, which makes centralized cloud computing face the challenge of computing and communication resources; (\textit{$\romannumeral2$}) Long-distance communication between cloud cluster and network devices can lead to high transmission latency, so it is difficult to meet the low latency requirements of services such as autonomous driving.

To address the above issues, edge computing~\cite{shi2016} and fog computing~\cite{yi2015survey} are expected to bring a new direction to the next-generation network architecture. They can be used to leverage device resources between the cloud and the end-user, thus combining the network edge in the cloud computing architecture to form the multi-tier edge-cloud computing system (following call it the multi-tier system for simplicity)\cite{yang2019multi, chen2018fog}. However, unlike the centralized cloud computing paradigm, the multi-tier system's widely distributed and heterogeneous system architecture poses new challenges\cite{tran2018}.

\subsection{Motivation and Challenges }
\label{subsec:Motivation and Challenges }

In cloud computing, cloud computing providers and application service providers propose Service-Level-Agreement (SLA) as a quantitative indicator to reflect service quality, so service management and optimization can be performed based on SLA\cite{wu2011, choi2021lazy}. Based on the idea of SLA, we provide service quality assurance for multiple services in the complex networked environment of multi-tier systems.

Although there are complementary advantages of various devices in a multi-tier system, and the overall efficiency and resource utilization of the system can be improved through the collaboration of various devices, there are three inherent challenges to be faced in the collaboration of various devices.
\vs{(\textit{$\romannumeral1$}) \textit{\textbf{Multi-resource heterogeneity}}. This challenge is in terms of the resource supply of edge devices. Due to the distributed deployment of edge clusters, the computing, caching and other resources provided by different edge clusters are highly heterogeneous, which makes it difficult to efficiently integrate and utilize edge cluster resources.
(\textit{$\romannumeral2$}) \textit{\textbf{Resource competition}}. This challenge is in terms of the resource requirements of application requests. Due to the resource limitation of a single edge cluster, all kinds of application services cannot be deployed in each edge cluster. Moreover, the resource requirements of different types of applications change dynamically in real-time, which makes it difficult to coordinate the resource competition among different application requests.
(\textit{$\romannumeral3$}) \textit{\textbf{Networked system dynamics}}. This challenge is in terms of matching the demand and supply of resources. Since requests need to be sent over the network to the edge devices to get the required resources, but the resources available in the system are always changing dynamically, resulting in the requests not being efficiently transmitted to the appropriate devices.
The three challenges mentioned above correspond to the supply of edge devices, the requirement of application services, and the matching of supply and requirement of resources, which restrict the overall efficient operation of the multi-tier system.}
Therefore, there is an urgent need to design a framework for cross-tier optimization of services and requests in the multi-tier system, to better guarantee the SLA of various services.

\begin{figure*}[t]%
	\centering
	\includegraphics[width=1.0\linewidth]{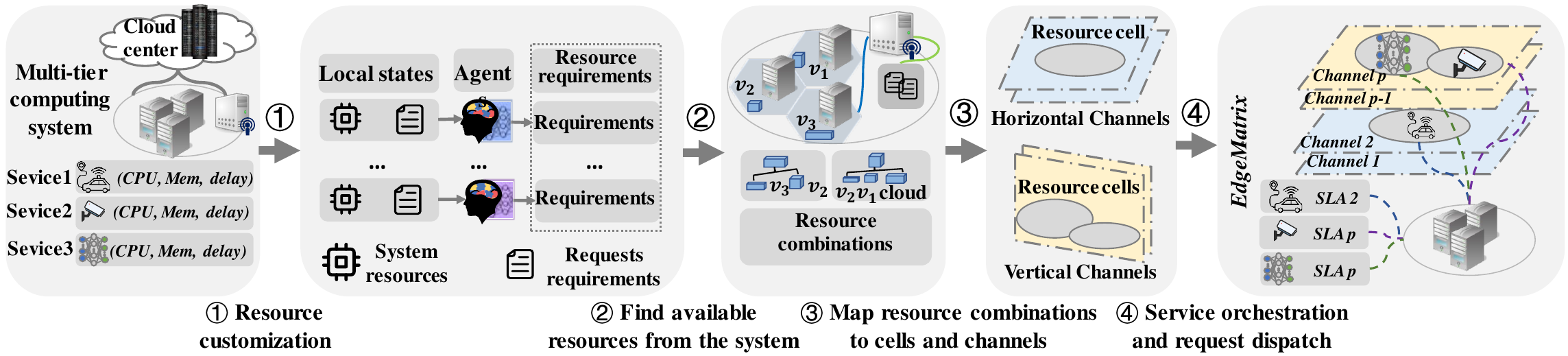}
	\setlength{\abovecaptionskip}{-0.5cm} 
	\caption{Understand \textit{EdgeMatrix} through a case study.}
	\label{fig:CaseStudy}
\end{figure*}

\subsection{Technical Challenges and Solutions}
\label{subsec:Technical Challenges and Solutions}

The multi-tier system is similar to a large company with a multi-tier organizational structure, and the relationship between the cloud center, fog node and edge node is just like the relationship between senior management, middle management and front-line employees in a company~\cite{yang2019multi}.
To better cope with the three inherent challenges in the multi-tier system, i.e., multi-resource heterogeneity, resource competition and networked system dynamics, our work correspondingly focuses on resource customization, service orchestration and request dispatch. Finally, we propose a joint optimization framework for multi-tier systems, i.e. \textit{EdgeMatrix}, which optimizes in the following aspects.

\textit{\textbf{Resource customization}}. \vs{The multi-tier system not only has a huge number of heterogeneous server nodes, but also has various types of services with different requirements. Therefore, the complexity of one-step decision algorithms based on global information will explode with increasing the number of server nodes and service types~\cite{wang2020service, jovsilo2020computation, shang2021deep}. Furthermore, the algorithm cannot make decisions flexibly in large-scale heterogeneous scenarios. To solve the above problem, we use multi-agent deep reinforcement learning (MADRL) to redefine heterogeneous physical resources so as to provide customized isolated resources (i.e., \textit{resource cell} and \textit{resource channel}). In this way, a coarse-grained optimization is performed using a low-complexity algorithm for different channels, and then a fine-grained optimization is performed in parallel for different cells in the channel. Thus, the impact of cluster size on algorithm complexity can be significantly reduced by two-step decision optimization. In addition, the impact of device heterogeneity can be reduced by adjusting the amount of cell and channel resources.} \textbf{Specially, we customize the resources of edge-edge nodes (Horizontal) and edge-cloud nodes (Vertical) to form logically isolated resource combinations called \textit{resource cells} in multi-tier systems.} 
In biology, cells in different locations have different functions and each cell has a separate space. Drawing on this biological concept, we name this reorganized resource unit as \textit{resource cells} to be the basic element of \textit{EdgeMatrix}.
\textbf{We further call the set of cells with similar characteristics (resources, latency, etc.) a \textit{resource channel}, which means that each resource channel has its corresponding SLA priority.}

\textit{\textbf{Service orchestration}}. \vs{Since the resources of a single edge node are usually not high, the number of services it can deploy is small. Therefore, it is difficult for a single edge node to meet the requirements of all kinds of services, which leads to resource competition among services \cite{zhang2019}. For example, a high latency tolerant service deployed in a cloud cluster can meet SLA requirements, but it is deployed in an edge node, leading to SLA violations for other latency-sensitive services due to insufficient resources at the edge.} Therefore, we need to reduce the negative impact of resource competition by orchestrating the services in \textit{EdgeMatrix}.

\textit{\textbf{Request dispatch}}. \vs{Both the available resources and the number of pending requests are highly dynamic, making it difficult to dispatch requests to the appropriate nodes for processing~\cite{hu2019}. In addition, the request processing needs to depend on the corresponding type of service, so the request dispatch and service orchestration need to be jointly optimized.}
Since requests need to be dispatched quickly to reduce queuing delays, but service orchestration requires large time intervals to avoid excessive service deployment costs, we adopt a two-time-scale framework, i.e., performing resource customization and service orchestration sequentially on a large time scale (frame) and request dispatch on a small time scale (slot).

\vs{In summary, the overview of \textit{EdgeMatrix} is shown in Fig.~\ref{fig:CaseStudy}. First, \textit{EdgeMatrix} specifies the resource requirements of the application and the state of physical resources in the multi-tier system~\ding{172}. After that, based on the solution shown in Sec.~\ref{subsec:MADRL for Resource Customization}, it finds the available resources in the system~\ding{173} and further combines them into cells and channels~\ding{174}. Finally, based on cells and channels, \textit{EdgeMatrix} use the solution as shown in Sec.~\ref{subsec:Joint Service Orchestration and Request Dispacth} to jointly optimize service orchestration and request dispatch~\ding{175}.}

\subsection{Main Contributions}

Some of the results of this paper have been presented in the conference version\cite{edgematrix}. Based on the previous work, this paper extends the work by refining system design and adding more experimental results. In summary, our main contributions are as follows:
\begin{itemize}[leftmargin=*]

	\item We design a Networked Multi-agent Actor-Critic (NMAC) algorithm to map physical resources in each region into logical resource combinations and improve system stability through off-line centralized training and online distributed decision-making.
	
	\item We propose a method based on mixed-integer linear programming (MILP) to solve JSORD and further reduce the solving time by performing a multi-task mechanism in parallel.
	
\begin{figure}[t]%
	\centering
	\includegraphics[width=1.0\linewidth]{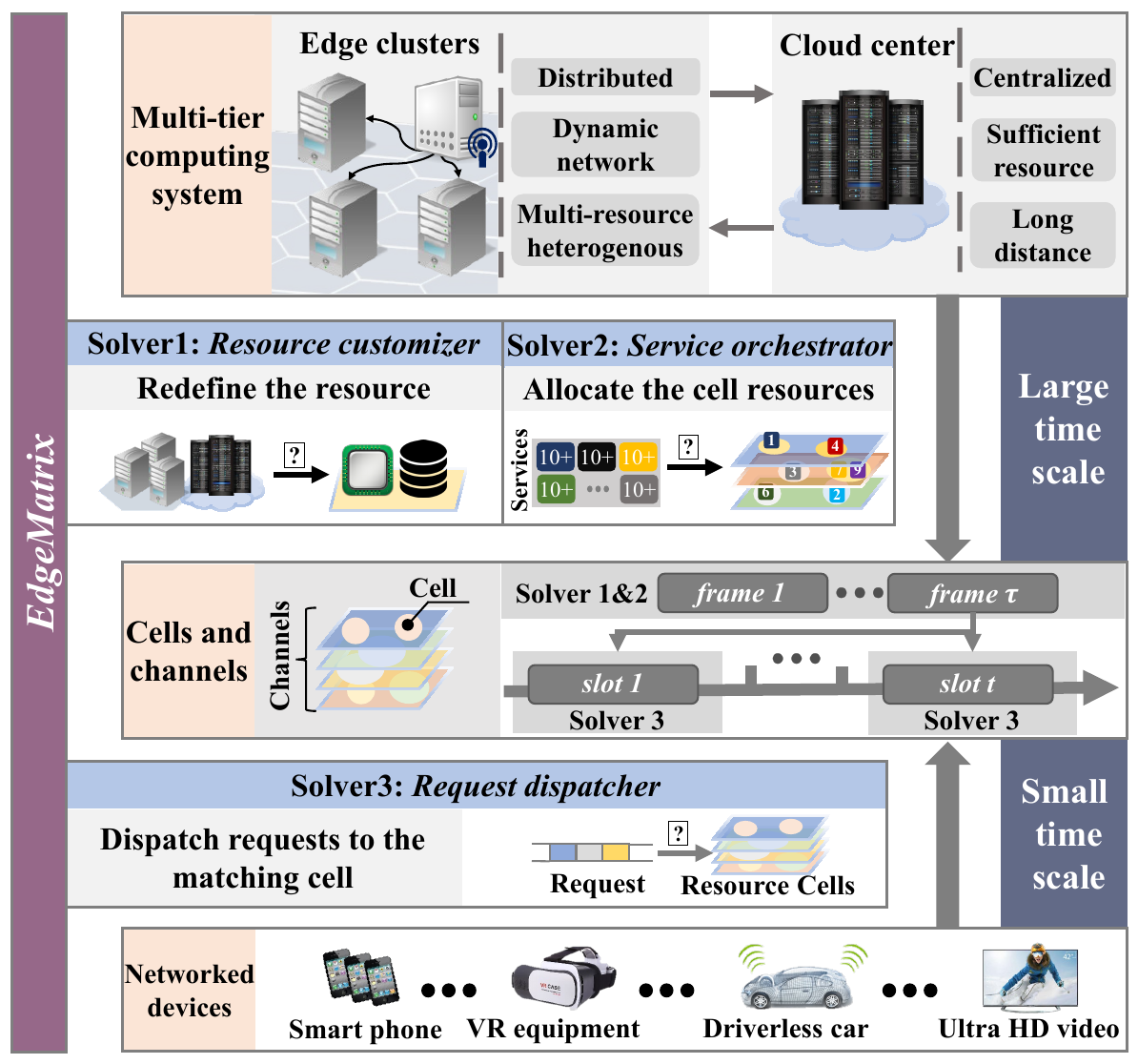}
	\setlength{\abovecaptionskip}{-0.5cm} 
	\caption{Resource customization with joint service orchestration and request dispatch in \textit{EdgeMatrix}.}
	\label{fig:Problem}
\end{figure}
	
	\item We design a two-time-scale framework to ensure coordinated operation of the components, that is, resource customization and service orchestration are performed in each frame, request dispatch is performed in each slot as shown in Fig. \ref{fig:Problem}, and the effectiveness of \textit{EdgeMatrix} is verified by evaluation based on real trace.

\end{itemize}

In the sequel, Sec. \ref{sec:System Model and Problem Statement} introduces the system model and state the problems. Sec. \ref{sec:Algorithm and System Design} and Sec. \ref{sec:System Design and Implementation} elaborate the algorithm and implementation. Sec. \ref{sec:Performance Evaluation} presents and analyzes the experiment results. Finally, Sec. \ref{sec:Related Work} reviews related works and Sec. \ref{sec:Conclusion} concludes the paper.

\section{System Model and Problem Statement}
\label{sec:System Model and Problem Statement}

\renewcommand{\arraystretch}{1.2} %
\begin{table*}[t]  
	\centering  

		\begin{threeparttable}  
			\caption{MAIN NOTATIONS}
			\label{table-1}  
			\begin{tabular}{m{1.0cm}m{6.8cm}m{1.0cm}m{7.1cm}}  
				\toprule          
				\multicolumn{1}{c}{\bf Notations }&\multicolumn{1}{c}{\bf Descriptions}&\multicolumn{1}{c}{\bf Notations }&\multicolumn{1}{c}{\bf Descriptions}\cr
				\midrule 
				
				$B_{i}$ & The total bandwidth of edge node $i$. & $R_{m_{i,\tau}}$ & The memory requirement of edge node $i$ at frame ${\tau}$.\cr
				$G_{d}(\mathcal{V}, \mathcal{E})$ & The network topology of the edge cluster in region $d$. & $R_{p,m}$ & The memory capacity of cell $m$ in channel $p$.\cr
				$h_{p,l}$ & The request packet size of service $l$ in channel $p$. & $t$ & The index of a slot.\cr
				$i$ & The index of an edge node. & $t_{i,m}$ & The transmission latency between node $i$ and cell $m$.\cr
				$\mathcal{L}_p $ & The set of service in channel $p$. & $t_{p,l}$ & The maximum response time of service $l$ in channel $p$.\cr
				$m$ & The index of a cell. & $w_{p,l}$ & The required computing capacity of service $l$ in channel $p$. \cr
				$\mathcal{M}_{p}$ & The set of cells in channel $p$. & $W_{\rm{cloud}}$ & The computing capacity owned by the cloud center.\cr
				$N$ & The number of edge nodes. & $W_{i}$ & The computing capacity of edge node $i$.\cr
				$\mathcal{N}_{i}$ & The set of edge node $i$ and its adjacent nodes. & $W_{m_{i,\tau}}$ & The computing requirement of edge node $i$ at frame ${\tau}$.\cr
				$o_{p,l}$ & The execution time of service $l$ in channel $p$. & $W_{p,m}$ & The computing capacity of cell $m$ in channel $p$.\cr
				$p$ & The index of a channel. & $x^{\tau}_{p,l,m}$ & Whether service $l \in \mathcal{L}_{p}$ is orchestrated on cell $m$ in frame $\tau$.\cr
				$\mathcal{P}$ & The set of SLA priorities or channels. & $y_{p,l,i,m}^{t}$ & The probability that a request of service $l \in \mathcal{L}_{p}$ arrived at edge node $i$ is dispatched to cell $m$ at slot t.\cr
				$r_{p,l}$ & The memory size required of service $l$ in channel $p$. & $\lambda^{t}_{p,l,i}$ & The number of requests from users for service $l \in \mathcal{L}_{p}$ arrived at node $i$ at slot $t$.\cr
				$R_{\rm{cloud}}$ & The total memory owned by the cloud center. & $\pi_{i,\tau}$ & The policy for resource customization. \cr
				$R_{i}$ & The total memory of edge node $i$. & $\tau$ & The index of a frame.\cr				
				
				\bottomrule  
			\end{tabular}  
		\end{threeparttable}
	 
\end{table*}

\subsection{Multi-tier System}

\begin{figure}[t]%
	\centering
	\includegraphics[width=1.0\linewidth]{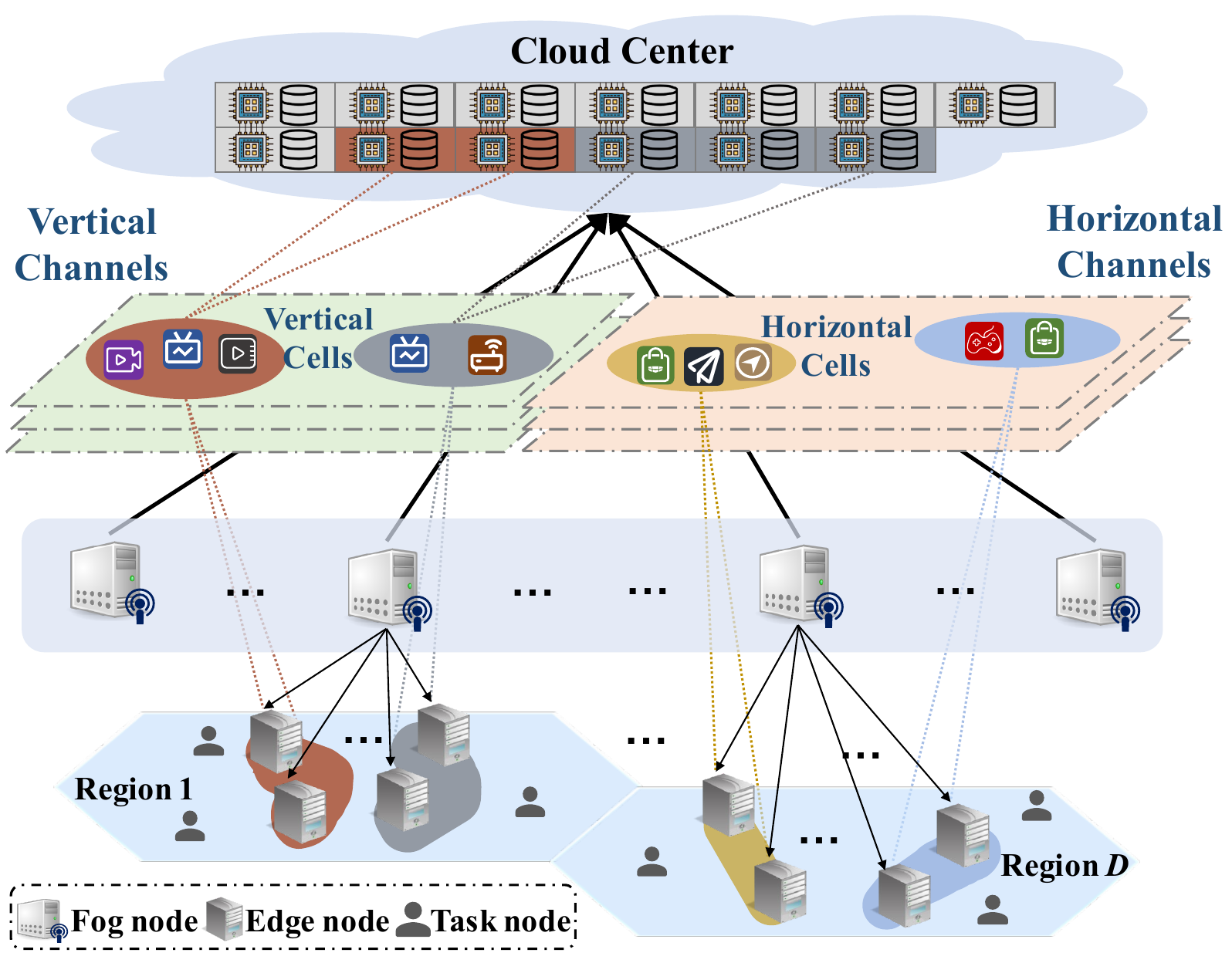}
	\setlength{\abovecaptionskip}{-0.5cm} 
	\caption{Architechture of the multi-tier system.}
	\label{fig:SystemArchitechture}
\end{figure}

In this paper, Table \ref{table-1} lists the main notations we will use. 
First of all, we consider a multi-tier system as shown in Fig. \ref{fig:SystemArchitechture}, which consists of the task node, edge node, fog node and the cloud center.
At the edge of the network, there exist massive heterogeneous edge nodes, and adjacent nodes in certain regions $\mathcal{D}=\{1,...,D\}$ together form edge clusters to provide closer resources. 
\vs{In addition, all edge nodes in each region are managed by the fog node, and the fog node can gather global information in the region to comprehensively manage and make decisions on resources, requests and services in the region.}
To introduce \textit{EdgeMatrix} concisely, we use multiple clusters in one region $d \in \mathcal{D}$ as an example, but \textit{EdgeMatrix} is applicable to other regions $d^\prime \in \mathcal{D}$. 
\vs{In addition, edge nodes and cloud center in a multi-tier system can deploy services to handle requests, but fog nodes only serve as regional decision centers without deploying services. However, \textit{EdgeMatrix} is also applicable for the application scenario where services are deployed in fog nodes, i.e., this part of fog nodes can be considered as edge nodes in this paper for resource redefinition.}

\textbf{\textit{Network edge}}. Geographically dispersed task nodes generate different arrival requests over time that have different SLA priorities $\mathcal{P}=\{1,...,P\}$, and each SLA $p \in \mathcal{P}$ has a service set$\mathcal{L}_{p}=\{1,...,L_{p}\}$. All services with different SLA priorities are denoted by $\mathcal{L}=\mathcal{L}_1\cup \mathcal{L}_2\cup ... \cup \mathcal{L}_P$. For service $l \in \mathcal{L}_p$, we denote the request packet size of each service $l$ as $h_{p,l}$, memory required of service $l$ as $r_{p,l}$, the required computing capacity of service $l$ as $w_{p,l}$, the maximum response time of services with $p$ (the lifecycle of service) as $t_{p,l}$, and the execution time of services with $p$ as $o_{p,l}$. 
We denote the network topology of edge clusters in region $d$ as a graph $G_{d}(\mathcal{V}, \mathcal{E})$, where each $i \in \mathcal{V}$ is the edge node, and $e_{ij} \in \mathcal{E}$ is the link directly connected between node $i$ and node $j$. $\mathcal{N}_{i}=\{j \mid j \in \mathcal{V}, e_{ij} \in \mathcal{E}  \}$ represents the neighborhood where the node $i$ is located, that is, the set of $i$ and its adjacent nodes. The number of edge nodes in cluster $G_{d}$ is denoted as $N$.
In addition, we denote edge node $i$ has the computing capacity $W_{i}$, the total memory $R_{i}$, the total bandwidth $B_{i}$.

\textbf{\textit{Cloud center}}. 
The centralized cloud center has more abundant resources than edge clusters. However, the centralized deployment also makes it difficult to respond quickly to distributed generated requests, resulting in high transmission latency for most requests. Therefore, it is often used to handle the types of tasks with high resource requirements but low latency requirements.

\textbf{\textit{Resource cells and channels}}. 
In order to isolate the negative impact of physical resource heterogeneity, \textit{EdgeMatrix} customizes the resources of different edge-edge devices (Horizontal) and edge-cloud devices (Vertical) into different resource cells. For different resource cells, they are logically independent from each other and can be mapped to a set of physical resources. In addition, service deployment in the \textit{EdgeMatrix} is assigned only to the resource cell instead of pointing the service to a physical device. Further, another concept is defined in \textit{EdgeMatrix}, namely resource channel, which consists of resource channels with similar characteristics. Based on the resource channel, \textit{EdgeMatrix} can better guarantee the SLA priority of services, i.e., provide resources with different characteristics for services with different SLA priorities.
Therefore, we can treat the channels and SLA of services as equivalent to $\mathcal{P}=\{1,...,P\}$. On each channel $p \in \mathcal{P}$, we deploy customized resource cells $\mathcal{M}_{p}=\{1,...,m_{p}\}$ for task nodes according to the SLA of services arriving in the edge cluster. For cell $m \in \mathcal{M}_{p}$, we denote its computing capacity as $W_{p,m}$, and memory size as $R_{p,m}$.

	\subsection{Resource Redefinition}

	\vs{In the multi-tier system, the available resources of edge devices can be better utilized, but a series of challenges are also faced.} In response, we address these challenges by redefining resources based on the multi-tier computing architecture at the physical level. As shown in Fig. \ref{fig:definition}, our definition and purpose of resource redefinition are as follows.

	\begin{figure}[t]%
		\centering
		\includegraphics[width=1.0\linewidth]{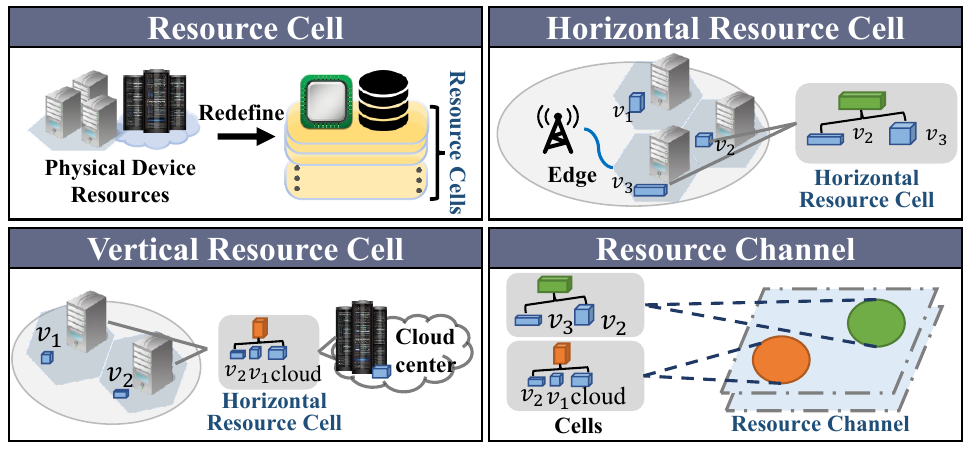}
		\setlength{\abovecaptionskip}{-0.5cm} 
		\caption{The description of the four types of redefined resources.}
		\label{fig:definition}
	\end{figure}
	
	\begin{itemize}[leftmargin=*]
		\item \textbf{\textit{Resource Cell}}. It is a resource set that \textit{EdgeMatrix} maps the physical resources in the multi-tier system. The heterogeneity of the physical resources of the devices makes it difficult to have uniform resource control at the device level. Therefore, the use of resource customizer (introduced in Sec. \ref{subsec:problem statement}) in the \textit{EdgeMatrix} provides a more fine-grained mapping of the physical resources of heterogeneous devices, resulting in a relatively unified resource set. In addition, by periodically adjusting the mapping of resource cells to physical resources through the resource customizer, \textit{EdgeMatrix} can flexibly utilise heterogeneous device resources and adaptively provide services with matching device resources.
	
		\item \textbf{\textit{Horizontal Resource Cell}}. It is a resource cell obtained by mapping the physical resources in multiple neighbouring edge nodes. Since all resources come from the edge close to the task node, it can better meet the latency requirements of the service. However, due to the resource limitations of the edge devices, the horizontal resource cell has a low transmission latency, but it has fewer resources available.
	
		\item \textbf{\textit{Vertical Resource Cell}}. It is a resource cell mapped from the physical resources from the edge clusters and the cloud center. Since some services that rely only on network edge are limited by insufficient resources to meet the requirements, the vertical resource cell can obtain enough physical resources from the cloud center. In this way, although the transmission delay of the service will increase, it can provide more abundant device resources.
	
		\item \textbf{\textit{Resource Channel}}. It is a collection of resource cells with similar characteristics. After mapping the resource cells on the basis of physical resources, \textit{EdgeMatrix} will classify them into several categories through clustering algorithms to further reduce SLA violations for different services. In this case, it can strengthen resource independence between services and reduce resource competition.

	\end{itemize}
\vs{

Based on the above, \textit{EdgeMatrix} can first map services to corresponding resource channels based on SLA priorities and achieve isolation between different services, so as to reduce latency through parallel optimization of each channel. After that, for each channel, JSORD is used to realize the co-optimization between cells, services and requests. Therefore, the system optimization is divided into two aspects by dividing channels and cells. First, the services are matched with the channels at the coarse-grained level to reduce the complexity of the multi-tier system, and then fine-grained collaboration is performed for different cells in each channel to realize the overall optimization of the multi-tier system.
}

\subsection{Problem Statement}
\label{subsec:problem statement}

The objective of \textit{EdgeMatrix} is to reduce SLA violations for various services while maximizing overall throughput. To ensure the robustness of \textit{EdgeMatrix}, we adopt a two-time-scale framework in the multi-tier system to realize resource customization, service orchestration and request dispatch.

At the large time scale, frame $\tau$, \textit{EdgeMatrix} performs two steps to guarantee the SLA priorities of different services: (\textit{$\romannumeral1$}) resource customization, which customizes the resources in the multi-tier system into resource cells according to the states of the system based on MADRL algorithm, and groups cells with similar characteristics into one resource channel using a clustering algorithm; (\textit{$\romannumeral2$}) service orchestration, which allocates the cell resources to service replicas and then binds the service replicas with allocated physical resources.

\begin{figure}[t]%
	\centering
	\includegraphics[width=1.0\linewidth]{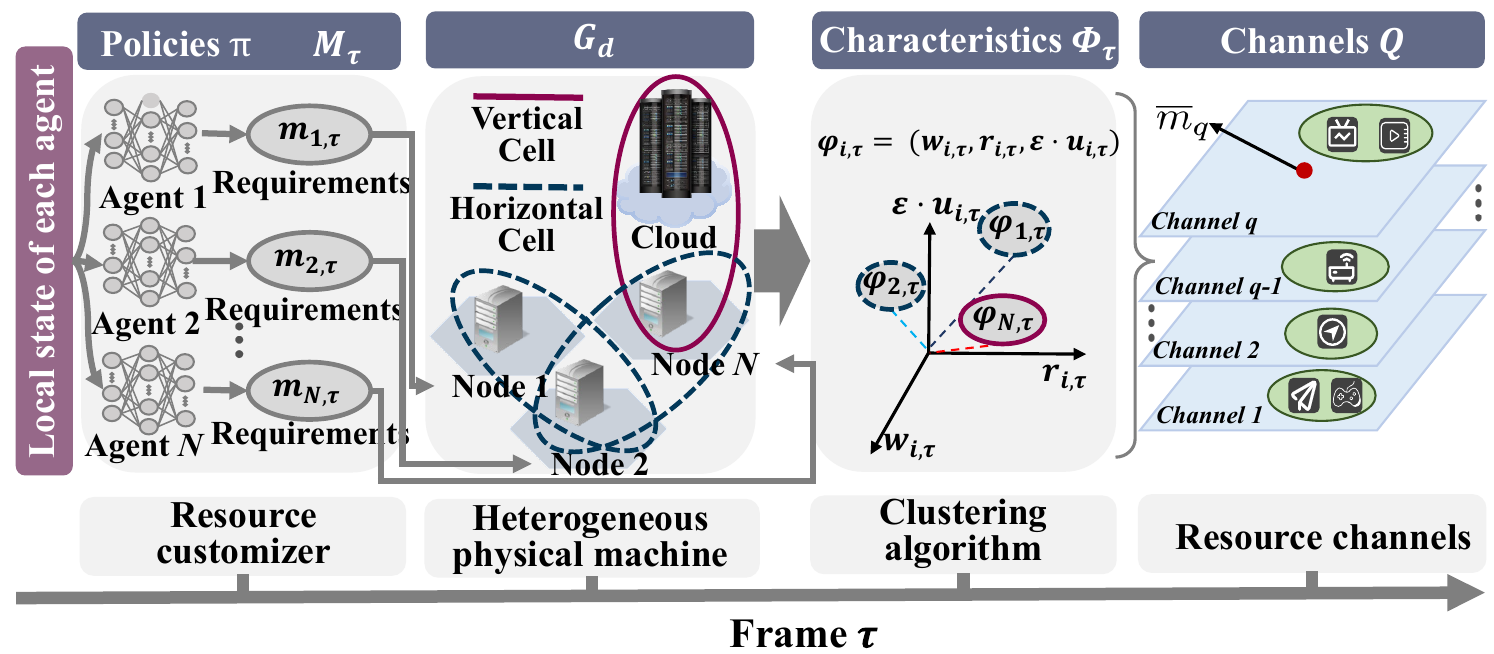}
	\setlength{\abovecaptionskip}{-0.5cm} 
	\caption{The workflow of \textit{resource customizer} at the large time scale.}
	\label{fig:The Workflow of resource customizer}
	\vspace{-0.6em}
\end{figure}

At the small time scale, namely slot $t$, \textit{EdgeMatrix} performs request dispatch to adapt the networked system dynamics. In \textit{EdgeMatrix}, the main components contain \textit{resource customizer}, \textit{service orchestrator} and \textit{request dispatcher}.

\textbf{\textit{Resource customizer}}. The workflow of \textit{resource customizer} is as shown in Fig. \ref{fig:The Workflow of resource customizer}. We deploy a \textit{resource customizer} agent for every node $i$ in the edge cluster. At frame ${\tau}$, all of the agents need to calculate new resource cell's resource requirements through the observed local state and learned resource customized policy $\pi_{i,\tau}$, denote by $M_{\tau}=\{m_{1,\tau},...,m_{i,\tau},...,m_{N,\tau}\}$. For $m_{i,\tau} \in M_{\tau}$, we denote the memory requirement as $R_{m_{i,\tau}}$ and the computing requirement as $W_{m_{i,\tau}}$. Note that, $M$ is different from $\mathcal{M}$ temporarily, each $m_{i,\tau} \in M_{\tau}$ only represents predicted resource requirements at this moment. The \textit{resource customizer} first obtains the available resources from the neighborhood $\mathcal{N}_i$ of the edge node $i$ where the agent is located according to the requirements of the resources required by $m_{i,\tau}$. When the available resources in the neighborhood can meet the requirements of $m_{i,\tau}$, we call it a horizontal resource cell. \textit{Resource Customizer} will obtain the rest resources if needed from the cloud center, namely the vertical resource cell.

After \textit{resource customizer} finishes customizing the physical resources in the multi-tier system into resource cells based on the MADRL algorithm, it has to cluster the resource cells to corresponding resource channels. First, the \textit{resource customizer} abstracts the characteristics from each cell, and then it uses a clustering algorithm to group the resource cells with similar characteristics to one resource channel. Finally, SLA priorities are defined for each channel to provide services to task nodes with a corresponding service level.
Horizontal resource cells have lower transmission latency but limited resources, while vertical resource cells have sufficient resources but higher transmission latency. We then group the cells with similar characteristics into the same channel, and the characteristics of resource cells in \textit{EdgeMatrix} is denoted by $\Phi_{\tau}=\{\varphi_{1,\tau}, ..., \varphi_{i,\tau},...\}$. Specifically, the characteristic of each resource cell $m_{i,\tau}$ is defined by $ \varphi_{i,\tau} = (w_{i,\tau}, r_{i,\tau}, \varepsilon\cdot u_{i,\tau}) $, where $w_{i,\tau}$ and $r_{i,\tau}$ are normalized CPU and memory resources, and $u_{i,\tau}$ is the edge resources proportion. 

This article mainly focuses on the two resources of CPU and memory. Note that the larger value of $u_{i,\tau}$ means the lower latency of $m_{i,\tau}$, and the latency is one of the essential factors affecting SLA priority, so we add the consideration of the weighting factor $\varepsilon$ to $u_{i,\tau}$. The larger the value of each item in $\varphi_{i,\tau}$, the better the performance of corresponding item will be in the resource cell $m_{i,\tau}$. We group the cells $m_{i,\tau} \in M_{\tau}$ based on $\Phi_{\tau}$ into resource channels  $p \in \mathcal{P}$ with corresponding SLA priority using clustering algorithm, i.e., $M_{\tau} \Rightarrow \mathcal{M}_{\mathcal{P}}$. The SLA priority of each channel is denoted by $ \delta_{p} = \sqrt{w_{\overline{m}_{p}}^{2}+r_{\overline{m}_{p}}^{2}+(\varepsilon\cdot u_{\overline{m}_{p}})^{2}} $, where $\overline{m}_{p}$ is the central point of channel $p$, and SLA priority is proportional to the performance of the resource cell.

\textbf{\textit{Service orchestrator}}. To take full advantage of the resources customized by \textit{resource customizer}, the \textit{service orchestrator} needs to appropriately orchestrate service replicas at resource cells of each resource channel. On a resource channel $p \in \mathcal{P}$, a service replica $l_{p}$ is orchestrated to the cell $m_{p}$, can be denoted as $(l,m)$. We define all sets of orchestration as $\mathcal{S} \subseteq \mathcal{L}_{p} \times \mathcal{M}_{p}$, and describe each individual service orchestration as selecting an element from the set. Thus, we can transform the service orchestration problem into a set optimization problem. Note that service orchestration significantly affects request dispatch, and we will explain the relationship between them later.

\textit{\textbf{Request dispatcher}}. After the resource customization and service orchestration is completed, \textit{request dispatcher} will dispatch the requests that reach the nodes to the resource cells with matching service replicas at the small time scale, slot $t$, as shown in Fig. \ref{fig:Section2_2_JSORD_Problem}. The number of requests from task nodes for service $l \in \mathcal{L}_{p}$ arrived at node $i$ at slot $t$ is $\lambda^{t}_{p,l,i}$, and the average number of requests for a frame $\tau$ is denoted by $\lambda^{\tau}_{p,l,i}$.

To achieve the objective of \textit{EdgeMatrix}, we have made the following efforts:  (\textit{$\romannumeral1$}) the policy learned by the \textit{resource customizer} reduces SLA violations for various services (detailed in Sec. \ref{subsec:MADRL for Resource Customization}); (\textit{$\romannumeral2$}) moreover, we jointly consider service orchestration and request dispatch to maximize the overall throughput of the system (i.e., JSORD); and model them together as a mathematical problem because of their strong correlation (a brief introduction in the following, and more details in the Sec. \ref{subsec:Joint Service Orchestration and Request Dispacth}).

Since each channel in $\mathcal{P}$ is similar in terms of handling JSORD, for clarity, one channel $p \in \mathcal{P}$ is used to introduce \textit{EdgeMatrix} in the following discussion. We first set up two decision variables $x$ and $y$, where $x$ is the service orchestration variable and $y$ is the request dispatch variable. More specifically, $x^{\tau}_{p,l,m} \in \left\{0,1\right\}$ is 1 if service $l$ is orchestrated on cell $m$ in frame $\tau$ and 0 otherwise, $y_{p,l,i,m}^{t} \in [0,1]$ represents the probability that a request of service $l$ arrived at edge node $i$ is dispatched to cell $m$ at slot t. For frame $\tau$, we define $y$ as $y_{p,l,i,m}^{\tau} \in [0,1]$. 

We formulate the JSROD as Eq. (1): The object of (1a) is to maximize the number of each channel served requests, $\Psi_{p}=\sum_{l\in\mathcal{L}_{p}}\sum_{i \in \mathcal{V}}\lambda_{p,l,i}\sum_{m\in \mathcal{M}_p}y_{p,l,i,m}$, which equivalent to the system overall throughput because the joint optimization among channels is mutually independent. Constraint (1b) guarantees the request dispatch variable is available. Constraints (1c) and (1d) ensure that each cell's memory and computing capacity can offer the resources required by orchestrated service replicas. Constraint (1e) ensures that $y$ is valid if and only when the service $l$ is orchestrated and won't trigger the SLA, where $\mathbb{I}_{t_{c,l}-o_{c,l}-t_{i,m} \textgreater 0}$ is the indicator function, indicates the SLA priority for the requests, and $t_{i,m}$ is the transmission latency between node $i$ and cell $m$. Constraint (1f) is the available range of values.

\begin{gather}
	{\rm max} \ \Psi_{p}, \qquad   \tag{1a} \\
	\textbf{{\rm s.t.}} \sum_{m\in \mathcal{M}_p} y_{p,l,i,m} \leq 1, \quad\qquad\qquad\qquad\qquad\qquad\quad\quad  \tag{1b}\\
	\ \sum_{l\in \mathcal{L}_p} x_{p,l,m} r_{p,l} \leq R_{p,m}, \quad\qquad\qquad\qquad\qquad\   \tag{1c}\\
	\ \sum_{l\in \mathcal{L}_p} \omega_{p,l} \sum_{i \in \mathcal{V}} \lambda_{p,l,i}y_{p,l,i,m} \leq W_{p,m}, \qquad\qquad\quad  \tag{1d}\\
	y_{p,l,i,m} \leq {\rm min}\{x_{p,l,m}, \mathbb{I}_{t_{p,l}-o_{p,l}-t_{i,m} \textgreater 0}\},  \quad\quad   \tag{1e}\\
	\ x\in \{0, 1\}, y \geq 0, \forall p \in \mathcal{P},\ l\in \mathcal{L}_p,\ i \in \mathcal{V},\ m\in \mathcal{M}_p. \tag{1f}
\end{gather}

\begin{figure}[t]%
	\centering
	\includegraphics[width=1.0\linewidth]{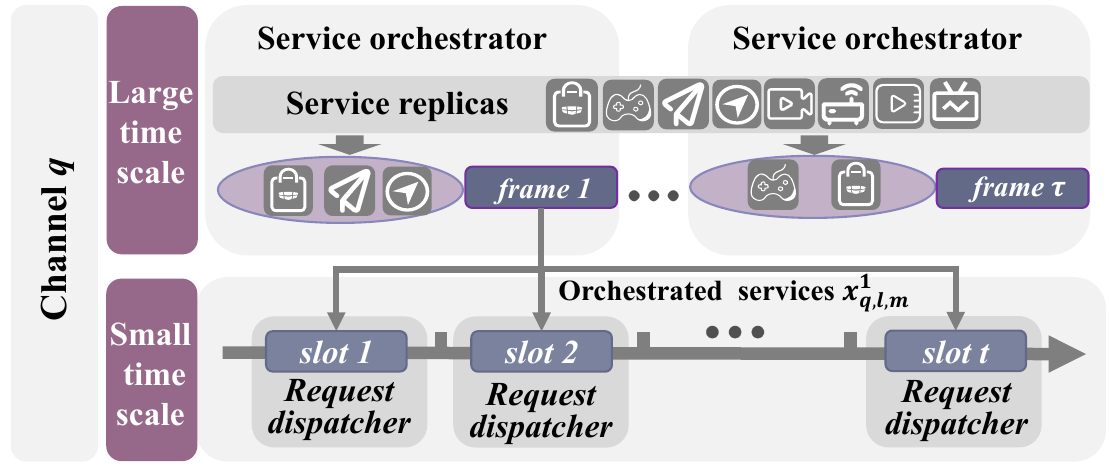}
	\setlength{\abovecaptionskip}{-0.5cm} 
	\caption{The workflow of \textit{service orchestrator} and \textit{request dispatcher} at the two-time-scale framework.}
	\label{fig:Section2_2_JSORD_Problem}
\end{figure}

\section{Algorithm Design}
\label{sec:Algorithm and System Design}

\subsection{MADRL for Resource Customization}
\label{subsec:MADRL for Resource Customization}
The \textit{resource customizer} agent deployed on each node can calculate new resource cells' resource requirements for \textit{EdgeMatrix} based on local state and learned policy. Its objective is to provide customized resource cells for services with different SLA priority under (\textit{$\romannumeral1$}) multi-resource heterogeneous edge nodes and (\textit{$\romannumeral2$}) dynamically changing service requests, which reduces SLA violations for various services.

With the development of artificial intelligence, especially reinforcement learning (RL)\cite{sutton2018} represented by DQN\cite{mnih2015}, DDPG\cite{timothy2016ddpg} and A3C\cite{mnih2016}, game control and robot control are well performed by RL. Due to the large number of computational nodes distributed in the multi-tier system, the direct implementation of these algorithms will cause a high-dimensional action space or non-stationary environment\cite{lowe2017multi,chu2019}. Therefore, we introduce the MADRL algorithm to enable each decision-capable edge node in the system to customize the resources in its network neighborhood into resource cells based on the changing system state. To learn practical resource customization policies in complex networked environments, we must consider (\textit{$\romannumeral1$}) the impact of algorithm training on the robustness of the networked system, (\textit{$\romannumeral2$}) the unsuitability for edge nodes with limited computational power to deploy large models, and (\textit{$\romannumeral3$}) the high-dimensional action space in the decision making of the networked system. Therefore, we adopt an algorithmic framework of offline centralized training and online distributed execution with a continuous action space.

\subsubsection{Markov Game Formulation}

Since the edge cluster of each region $d \in \mathcal{D}$ in \textit{EdgeMatrix} is a graph $G_{d}$, a multi-agent Markov Decision Process (MDP) can be formed as $\rho=(G_{d},\{\hat{\mathcal{S}_{i}},\hat{\mathcal{A}_{i}}\}_{i \in \mathcal{V}},\hat{\mathcal{P}}, \{\hat{\mathcal{R}_{i}}\}_{i \in \mathcal{V}})$. We denote the resource customizer agent on each edge node as $i \in \mathcal{V}$. More details are introduced according to $\rho$ in the following. 

\textbf{\textit{State space $\hat{\mathcal{S}}$}}. At frame $\tau$, the local state space observed by the agent on edge node $i$ is $\hat{s}_{i,\tau}$, which contains: (\textit{$\romannumeral1$}) the number and kinds of requests $(\lambda_{1,1,i}^{\tau}, ..., \lambda_{p,l,i}^{\tau})_{p\in\mathcal{P},l\in\mathcal{L}_{p}}$; (\textit{$\romannumeral2$}) the resource requirements and delay demand of requests arrived at node $i$; (\textit{$\romannumeral3$}) the CPU, memory and edge resources proportion of the existing cells created by agent $i$ in the system; (\textit{$\romannumeral4$}) the available resources of edge nodes where agent $i$ is located, $\mathcal{N}_{i}$. We simply consider the global observation for the training critic as the ensembles of all agents' state, $\bm{\hat{s}}_{\tau}$.

\textbf{\textit{Action space $\hat{\mathcal{A}}$}}. We define the action space of all agents in the edge cluster as a joint action space $\hat{\mathcal{A}}=\{\hat{\mathcal{A}}_{1},...,\hat{\mathcal{A}}_{i},...,\hat{\mathcal{A}}_{N}\}_{i \in \mathcal{V}}$ , where $\hat{\mathcal{A}}_{i} \in \hat{\mathcal{A}}$ represents the action space of agent $i$. At frame $\tau$, agent $i$ predict the action $\hat{a}_{i,\tau}$ according to the observed local state space $\hat{s}_{i,\tau}$ and policy $\pi_{i,\tau}$. Specifically, $\hat{a}_{i,\tau}$ indicates the size of the resource that the agent $i$ predicts allocate to cell $m_{i,\tau}$, i.e., $(W_{m_{i,\tau}}, R_{m_{i,\tau}})$, where both of them are continuous variables with a value range of $[0, 1]$. The actual resource size is $(\alpha \cdot W_{m_{i,\tau}}, \beta \cdot R_{m_{i,\tau}})$, where $\alpha$ and $\beta$ are the upper limits of the cell's resources.

\textbf{\textit{Reward function $\hat{\mathcal{R}}$}}. Agent $i$ inputs the observed local state space $\hat{s}_{i,\tau}$ and selected action $\hat{a}_{i,\tau}$ at frame $\tau$ into the reward function $\hat{\mathcal{R}}$ to get an immediate reward $\hat{r}_{i,\tau}$. To learn how to improve the overall throughput of the system while reducing SLA violations for various services, we comprehensively consider service throughput and SLA priority to help the agent learn this ability in an environment where multi-agents coordinate with each other. The reward function can be formulated as $\hat{r}_{i,\tau}=\sum_{p\in\mathcal{P}}\delta_{p}\sum_{l\in\mathcal{L}_{p}}\Psi'_{l,i,\tau}$, where $\Psi'_{l,i,\tau}=\Psi_{l,i,\tau}/\lambda_{p,l,i}^{\tau}$ is the throughput rate of service $l$ arrived at node $i$ at $\tau$. $\delta_{p}$ indicates the weight of the services with SLA priority $p$, and we verified the necessity of this setting by adjusting $\mathcal{\varepsilon}$ in Sec. \ref{subsec:Setting of Key Parameters}.

\textbf{\textit{State transition function $\hat{\mathcal{P}}$}}. Note that we use a deterministic policy, and the state transition function is denoted as $\hat{\mathcal{P}}(\hat{\mathcal{S}}' \mid \hat{\mathcal{S}},\hat{\mathcal{A}}_{1},...\hat{\mathcal{A}}_{N}): \hat{\mathcal{S}} \times \hat{\mathcal{A}}_{1} \times \ldots \times \hat{\mathcal{A}}_{N} \mapsto \hat{\mathcal{S}}'$.
\begin{figure}[t]%
	\centering
	\includegraphics[width=1.0\linewidth]{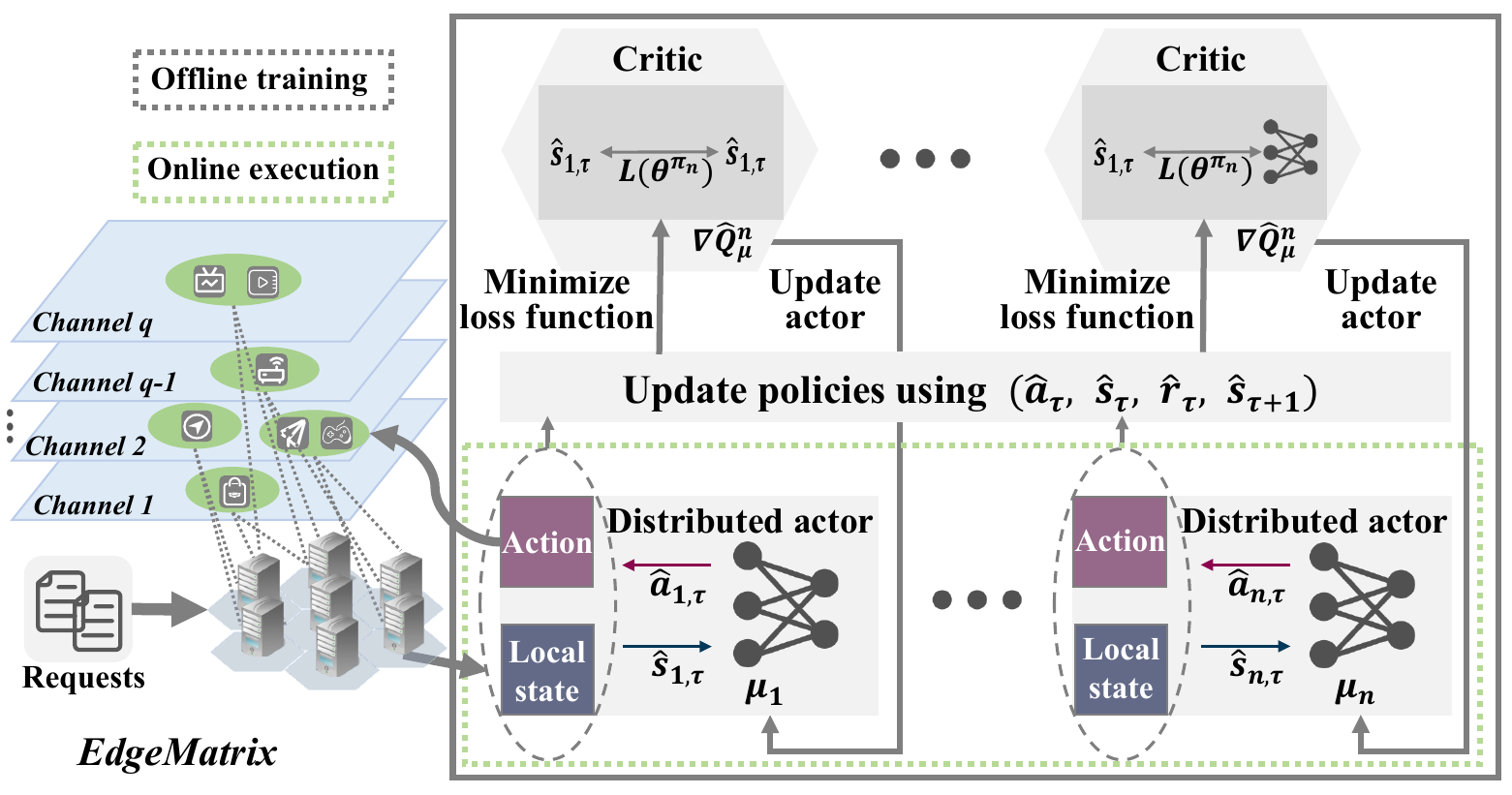}
	\setlength{\abovecaptionskip}{-0.5cm} 
	\caption{Networked multi-agent actor-critic.}
	\label{fig:NMAC}
	\vspace{-0.6em}
\end{figure}

\subsubsection{Networked Multi-agent Actor-Critic}
\label{subsubsec:Coordinated Multi-agent Actor-Critic}
In the networked environment of the multi-tier system, the main problems in designing the MADRL algorithm are: (\textit{$\romannumeral1$}) we need to minimize the requirements of resources for agents' decision-making because the computing nodes at the edge of the network only have limited resources; (\textit{$\romannumeral2$}) reducing SLA violations is the aim of our work, so the training or execution process should not affect the stability and security of the networked system. 

To better deal with the above problems, we proposed the \textit{Networked Multi-agent Actor-Critic (NMAC)} algorithm in the multi-agent coordinate environment of multi-tier system, as shown in Fig. \ref{fig:NMAC}, (\textit{$\romannumeral1$}) centralized critic, which can guide each actor to learn an effective policy according to global observation with extra information during training; (\textit{$\romannumeral2$}) distributed actor, each actor's input during training and execution is local state, so the actor can seamlessly switch between the two phases.
\begin{itemize}[leftmargin=*]
	
	\item \textbf{\textit{Centralized critic}}. During the training process, we equip each agent with a critic to train the actor. For agent $i$, the critic is implemented based on the centralized action-value function $\hat{Q}\left(\bm{\hat{s}}_{\tau}, \bm{\hat{a}}_{\tau} \mid \theta^{\pi_{i}}\right)$, which represents the expected discounted cumulated reward of frame $\tau$ starting from state-action pairs $(\bm{\hat{s}}_{\tau}, \bm{\hat{a}}_{\tau})$ according to the policy $\pi_{i}$, $\bm{\hat{a}}_{\tau}=(...,\hat{a}_{i,\tau},...)$. The action-value function can be represent as $\hat{Q}\left(\bm{\hat{s}}_{\tau}, \bm{\hat{a}}_{\tau} \mid \theta^{\pi_{i}}\right)=\mathbb{E}_{\pi_{i}}\left[R_{i,\tau}\right]$, where $R_{i,\tau}=\hat{r}_{i,\tau}+\sum_{\tau'=\tau+1}^{\mathcal{T}}\gamma^{(\tau'-\tau)}\hat{r}_{i,\tau+1}$. Thus, the centralized action-value function can be obtained from the Bellman Equation:
	\begin{equation}
		\hat{Q}\left(\hat{\bm{s}}_{\tau}, \hat{\bm{a}}_{\tau} \mid \theta^{\pi_{i}}\right)=\hat{r}_{i,\tau}+\gamma \max \limits_{\hat{\bm{a}}_{\tau+1}} \hat{Q}\left(\hat{\bm{s}}_{\tau+1}, \hat{\bm{a}}_{\tau+1} \mid \theta^{\pi_{i}}\right),\tag{2}
	\end{equation}
	\vs{where $\theta^{\pi_{i}}$ is the parameter of policy $\pi_{i}$, $\gamma$ is the reward discount factor. Then, the optimization function can be derived as a loss function between target critic network $\hat{Q}^{\prime}\left(\hat{\bm{s}}_{\tau}, \hat{\bm{a}}_{\tau} \mid \theta^{\pi'_{i}}\right)$ and actual critic network $\hat{Q}_{i,\tau}$, as follows:
	\begin{equation}
		\label{equation:Critic update function}
		\begin{array}{l}
			L\left(\theta^{\pi_{i}}\right)=\mathbb{E}\left[\left(\hat{Q}\left(\hat{\bm{s}}_{\tau}, \hat{\bm{a}}_{\tau} \mid \theta^{\pi_{i}}\right) - \hat{Q}^{\prime}\left(\hat{\bm{s}}_{\tau}, \hat{\bm{a}}_{\tau} \mid \theta^{\pi'_{i}}\right)\right)^{2}\right].\tag{3}
		\end{array}
	\end{equation}}

	\item \textbf{\textit{Distributed actor}}. For each agent, the actor network learns a deterministic policy $\mu_{i}$ to maximize the cumulative reward, i.e., $J=\mathbb{E}_{\mu_{i}}\left[R_{i,\tau}\right]$. We update the parameters $\theta^{\mu_{i}}$ through optimizing policy gradient:
	\begin{multline}
		\label{equation:Actor update function}
		\nabla_{\theta^{\mu_{i}}} J\left(\mu_{i}\right)=\\
		\mathbb{E}\left[\nabla_{\theta^{\mu_{i}}} \log \mu_{i}\left(\hat{a}_{i,\tau} \mid \hat{s}_{i,\tau}\right) \nabla_{a_{i}}\hat{Q}\left(\hat{s}_{i,\tau}, \hat{a}_{i,\tau}\mid\theta^{\mu_{i}}\right)\right]. \tag{4}
	\end{multline}
\end{itemize}

Especially, \textit{NMAC} implements an offline training and online execution framework: (\textit{$\romannumeral1$}) offline-training, can avoid that the training process may have a negative impact on the networked system stability; (\textit{$\romannumeral2$}) online-execution, only requires the actor-network to predict the action and the learned policy only uses local state, which significantly reduces the resources consumed by the agent compared to the training phase.

\begin{algorithm}[t]
	\caption{Solve JSORD Based on Submodular Function Maximization}
	\label{alg:JSORD}
	\KwIn{Input parameters of Eq. (1)\;} %
	\KwOut{Service orchestration variable $x^{\tau}_{p,l}$ and requests dispatch variable $y_{p,l,i,m}^{\tau}$\;}
	
	Initialize  $frame =\tau,\ channel=p$, $\mathcal{S} = \emptyset $, $T = \{e | e \in (\mathcal{L}_p\times \mathcal{M}_p) \setminus \mathcal{S}$, $\mathcal{S}\cup\{e\}$ satisfies constraints of (\ref{equation:7})$\}$ \;
	\While{$ T \neq \emptyset $}  
	{  
		$e^* = $ the element $e$ in $T$ that get the maximum value of \textit{\textbf{$\Omega(\mathcal{S}\cup\{e\}$}}\;
		$\mathcal{S} = \mathcal{S}\cup \{e^*\}$\;
		$T = \{e | e \in (\mathcal{L}_p\times \mathcal{M}_p) \setminus \mathcal{S}$, $\mathcal{S}\cup\{e\}$ satisfies constraints of (\ref{equation:7})\}\;
	}
	Convert $\mathcal{S}$ to its vector representation $x^{\tau}_{p,l}$\;
	Compute $y_{p,l,i,m}^{\tau} = \{..., y_{p,l,i,m}^{t}, ...\}$ using Eq. (\ref{equation:6}) based on orchestrated services $x^{\tau}_{p,l}$\;
	\For{slot $t=0,1,2,...,$}
	{
		Execution request dispatch with dispatching variable $y_{p,l,i,m}^{t}$ at the small time scale.
	}

\end{algorithm}

\subsection{Joint Service Orchestration and Request Dispacth}
\label{subsec:Joint Service Orchestration and Request Dispacth}

Since service orchestration significantly impacts request dispatch, we together consider them as a joint optimization problem, i.e., JSORD. Specifically, (\textit{$\romannumeral1$}) at the large time scale, JSORD has to orchestrate the appropriate services for each cell in \textit{EdgeMatrix} based on the system state, and (\textit{$\romannumeral2$}) at the small time scale, JSORD has to dispatch the requests arriving at each node of the system to resource cells. Considering the system with multiple types of constraints, (e.g., computation, memory, communication, and latency requirements), we solve JSORD based on MILP as shown in Algorithm~\ref{alg:JSORD}. However, the widely distributed edge nodes and the variety of services in the system make the runtime of the solution unacceptable for task nodes.

In \textit{EdgeMatrix}, each channel serves a specific class of services with the same SLA priority. The resources used during service orchestration and request dispatch are the cells of the channel, so the channels are independent of each other. Therefore, we execute JSORD independently in parallel on each channel, which significantly reduces the runtime.

\begin{figure}[t]%
	\centering
	\includegraphics[width=1.0\linewidth]{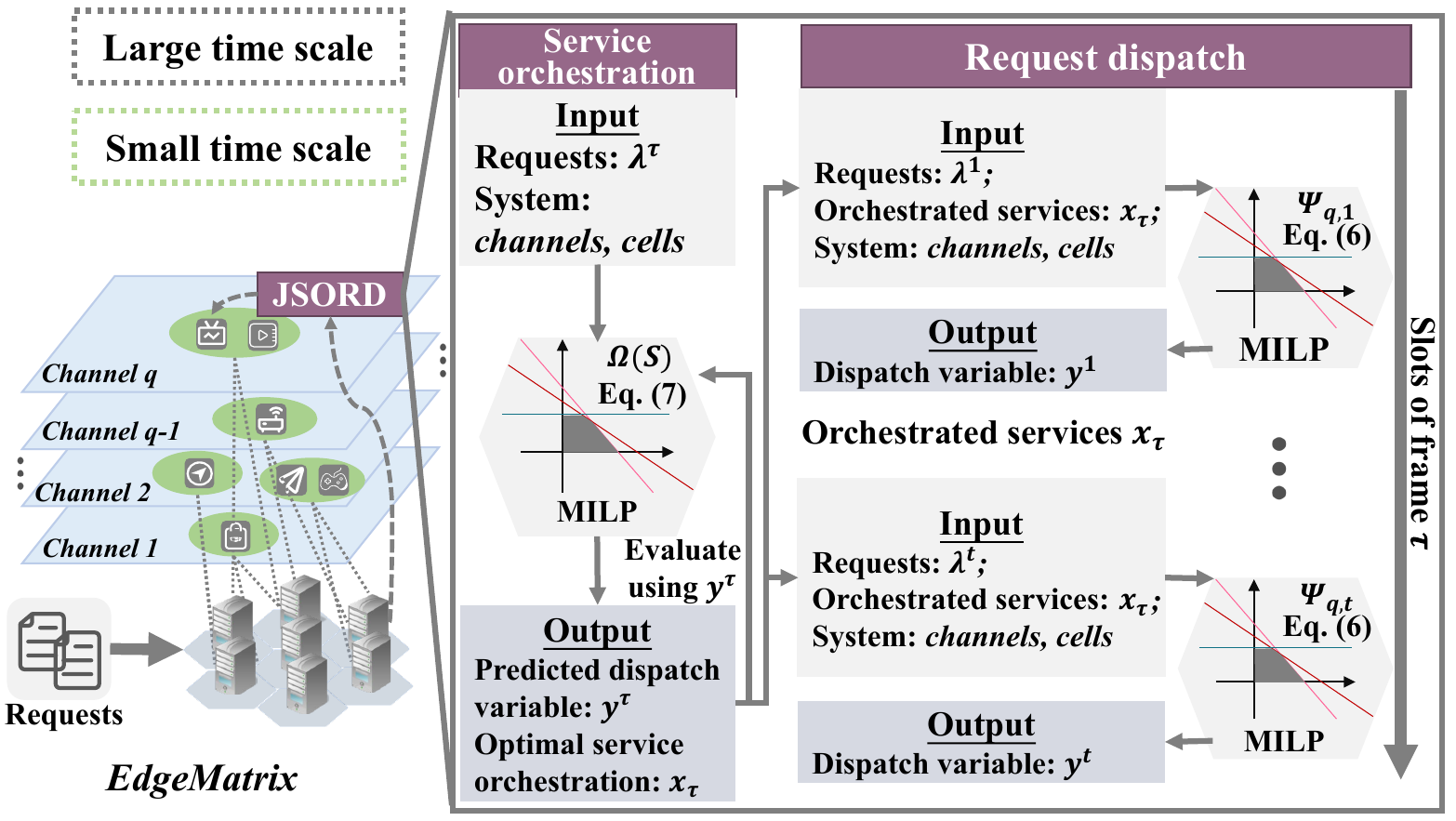}
	\setlength{\abovecaptionskip}{-0.5cm} 
	\caption{Joint service orchestration and request dispatch.}
	\label{fig:JSORD}
	\vspace{-0.6em}
\end{figure}

We model the joint optimization problem of service orchestration and request dispatch as Eq. (1) at the end of section \ref{sec:System Model and Problem Statement}. More specifically, (\textit{$\romannumeral1$}) at the beginning of frame $\tau$, we calculate the optimal service orchestration $x_{p,l}^{\tau}$ based on Eq. (1) by the predicted request dispatching probability $y_{p,l,i,m}^{\tau}$ and the request demand $\lambda_{p,l,i}^{\tau}$; (\textit{$\romannumeral2$}) then at the beginning of slot $t$, we calculate the request dispatch variable $y_{p,l,i,m}^{t}$ with the current request demand $\lambda_{p,l,i}^{t}$ and the orchestrated service $x_{p,l}^{\tau}$ to solve Eq. (1). At the large time scale, note that although the dispatch variable $y_{p,l,i,m}^{\tau}$ is predicted according to the demand $\lambda_{p,l,i}^{\tau}$, $y_{p,l,i,m}^{\tau}$ is used to evaluate Eq. (1a) under the given orchestrated services $x_{p,l}^{\tau}$ rather than request dispatch.

\subsubsection{Solvability Analysis}
We first analyze the solvability of the joint optimization Eq. (1) and then consider the special case of Eq. (1), where the resource cell and the service are homogeneous, with constraints (1d) ignored: %
\begin{gather}
	{\rm max} \ \Psi_{p},  \qquad\qquad\quad \tag{5a}\\
	\textbf{{\rm s.t.}} \ (1b),(1e),(1f), \qquad\qquad\qquad\quad\tag{5b}\\
	\sum_{l\in \mathcal{L}_p} x_{p,l,m} \leq R_{p, m}. \qquad\quad\quad \tag{5c}
\end{gather}

The joint optimization Eq. (1) can be simplified to the 2-Disjointed Set Cover Problem, i.e., Eq. (5), which is proved to be NP-complete\cite{cardei2005}. The special case of the joint optimization, Eq. (1) is NP-hard, which means that the joint optimization problem of service orchestration and request dispatch is also NP-Hard in the general case.

To describe the solution process of Eq. (1) more clearly, we discuss the deformations of Eq. (1) and the solution processes at the two-time-scale framework in the following, i.e., service orchestration and request dispatch, as shown in Fig. \ref{fig:JSORD}.

\subsubsection{Requst Dispatch}
\label{subsubsec:GNN-based System State Encoding}

At slot $t$, the \textit{service orchestrator} has already orchestrated the services on each resource cell, which means we solve the request dispatch problem under the situation that service orchestration variable $x_{p,l}^{\tau}$ is known. Thus, the joint Eq. (1) can be simplified to linear programming, i.e., Eq. (6), which means we can get the probability $y_{p,l,i,m}^{t}$ when dispatching a request of service $l$ arriving at the edge node $i$ to cell $m$. (If the request is successfully dispatched with our constraints (6b)-(6d), we can serve it.)
\begin{gather}
	\label{equation:6}
	{\rm max} \ \Psi_{p},  \qquad\qquad  \tag{6a}\\
	\textbf{{\rm s.t.}} (1b), (1d), (1e),(1f), \qquad\qquad  \tag{6b}\\
	\qquad\  y_{p,l,i,m} \leq \mathbb{I}_{(l,m)\in S}, \qquad\quad\   \tag{6c}\\
	\qquad\  y_{p,l,i,m} \in [0, 1]. \qquad\qquad\quad \tag{6d}
\end{gather}

\subsubsection{Approximation Algorithm for Service Orchestration}

Service orchestration problem can be transformed to a set optimization problem in the Sec. \ref{subsec:problem statement}. On resource channel $p \in \mathcal{P}$, orchestrating a service replica $l_{p}$ to the cell $m_{p}$ can be denoted as $(l,m)$. We can define all orchestration sets as $\mathcal{S} \subseteq \mathcal{L}_{p} \times \mathcal{M}_{p}$, and each single service orchestration can be described as selecting an element from the set. Let $\Omega\left(\mathcal{S}\right)$ denote the optimal objective value of Eq. (1) under a fixed set $\mathcal{S}$ of orchestrated services and a fixed dispatch variable $x$, $(l,m) \in \mathcal{S}$ if and only if $x_{l,m}=1$. This can be calculated by solving the request dispatch problem (see Eq. (6)), and then we can rewrite the problem as:

\begin{algorithm}[t]
	\caption{The Overall Algorithm of \textit{EdgeMatrix}}
	\label{alg:OverallAlgorithm}
	Initialize the system environment\;
	Initialize training parameters\;
	Get the system observation $\hat{s}_{0}$\;
	\For{frame $\tau=0,1,2,...,$}
	{
		Get the actions of each agent $\hat{a}_{\tau}=\left(\hat{a}_{1,\tau},...,\hat{a}_{i,\tau},...\right)$\;
		Resource Customizer execution actions to get the redefined resources, cells and channels\;
		\For{channel $p=0,1,2,...,$}
		{
			Solve \textbf{JSORD} based on Algorithm\ref{alg:JSORD} in parallel on each channel $p\in\mathcal{P}$ \;
		}
		Get reward $\hat{r}_{\tau}$ and next observation $\hat{s}_{\tau+1}$\;
		Store sample data [$\hat{s}_{\tau},\hat{a}_{\tau},\hat{r}_{\tau},\hat{s}_{\tau+1}$] for updating neural network\;
		\If{frame \% update rate $==0$}
		{
			Update the parameters of actor ($\theta^{\mu}$) and critic($\theta^{\pi}$) using Eq. (\ref{equation:Critic update function}) and Eq. (\ref{equation:Actor update function})\;
			Save models.
		}
	}
\end{algorithm}

\begin{gather}
	\label{equation:7}
	{\rm max} \ \Omega(\mathcal{S}), \quad\qquad\  \tag{7a} \\
	\textbf{{\rm s.t.}} \sum_{l:(l, m)\in \mathcal{S}}r_{p,l} \leq R_{p,m},    \qquad \tag{7b}\\
	\qquad \mathcal{S}\subseteq \mathcal{L}_p\times \mathcal{M}_p. \qquad\ \  \tag{7c}
\end{gather}
In summary, the overall training and scheduling process of \textit{EdgeMatrix} is given in Algorithm \ref{alg:OverallAlgorithm}.

	\subsection{Adaptability to Multi-tier System}

	\vs{The multi-tier system has a higher dimensional state space than centralized cloud computing, which makes the problem more challenging.} First, the physical resources of devices are distributed and highly heterogeneous, making the management of resources face greater complexity. In response, \textit{EdgeMatrix} reorganizes the physical resources as resource cells through resource customization and defined channels according to the SLA priority of the service. Therefore, \textit{EdgeMatrix} can isolate the impact of heterogeneous physical resources by redefining physical resources, thus realizing cross-tier management of heterogeneous resources. In addition, \textit{EdgeMatrix} uses MILP to solve JSORD based on redefined resources and enhances the SLA guarantee of services through channels to realize cross-tier management of services and requests. In summary, \textit{EdgeMatrix} can jointly optimize devices, services and requests in multi-tier systems across tiers.
	
\vs{In terms of practical prospects, computing resource providers such as AWS~\cite{Aws_mec} and Azure~\cite{Azure_mec} are converging servers at the network edge based on cloud computing and gradually building commercial multi-tier computing systems. In addition, the idea of Serverless has attracted a lot of attention~\cite{serverless-aws, serverless-cf}, which will make YouTube, Uber and other application service providers only focus on the application itself, while computing resource providers such as AWS and Amazon will design optimization frameworks to ensure the efficient operation of services. Therefore, \textit{EdgeMatrix} will be applicable to computing resource providers, which can help them manage heterogeneous physical resources, orchestrate services from different application service providers, and dispatch user requests associated with services.}
	\subsection{Theoretical Analysis}
	Next, we prove the approximation ratio of the algorithm, before which we need to show that the objective function of Eq. (7) is monotone and sub-modular~\cite{fisher1978}.

	\begin{definition}
		\textit{A set function $f$: $2^{\rm{x}} \rightarrow \mathcal{R}$ is monotone increasing if $\forall \mathcal{S}_1 \subseteq \mathcal{S}_2 \subseteq \rm{x}$, $f(\mathcal{S}_1 ) \leq f(\mathcal{S}_2 )$. Moreover, the function $f(.)$ is sub-modular if $\forall \mathcal{S}_1 \subseteq \mathcal{S}_2 \subseteq \rm{x}$ and $e \in\rm{x} \setminus \mathcal{S}_2$, $f(\{e\} \cup \mathcal{S}_1)- f(\mathcal{S}_1) \geq f(\{e\} \cup \mathcal{S}_2)- f(\mathcal{S}_2).$}
	\end{definition}

	\begin{lemma}
		\label{lemma1}
		\textit{The objective function in (\ref{equation:7}) is a monotone sub-modular function for all feasible $\mathcal{S}$ if }
		\begin{itemize}
			\item [1)]  \textit{$\lfloor R_{p,m}/r_{p,l} \rfloor \leq 1$ for all $m \in \mathcal{M}_p$ and $l \in \mathcal{L}_p$, or}
			
			\item [2)]  \textit{$\sum_{l\in \mathcal{L}_p} \omega_{p,l} \sum_{i \in \mathcal{V}} \lambda_{p,l,i} \leq W_{p,m}$ for all $m \in \mathcal{M}_p$.}
		\end{itemize}
	\end{lemma}

	\begin{proof}
		It is easy to judge that $\Omega(\mathcal{S})$ is monotone because adding an element to $\mathcal{S}$ will relax the constraint (6c), thereby expand the solution space for Eq. (6) and increase its optimal objective value.
		
		For proving the objective function of Eq. (7) is sub-modular, we need to show that for any sets $\mathcal{S}_1, \mathcal{S}_2 \subseteq \mathcal{L}_p \times \mathcal{M}_p$ and any $(l_1, m_1 ) \in (\mathcal{L}_p \times \mathcal{M}_p) \setminus \mathcal{S}_2$, such that $\mathcal{S}_1 \subseteq \mathcal{S}_2$ and $\mathcal{S}_2 \cup \{(l_1, m_1)\}$ is feasible, the following relationship holds
		\begin{equation}
			\label{equation:8}
			\Omega(\mathcal{S}_1 \cup \{(l_1, m_1)\}) - \Omega(\mathcal{S}_1) \geq \Omega(\mathcal{S}_2 \cup \{(l_1, m_1)\}) - \Omega(\mathcal{S}_2) \tag {8}
		\end{equation}

		We suppose  $\textbf{y}^{(0)}$ and $\textbf{y}^{(2)}$ are the optimal dispatching solutions according to Eq. (6) under service orchestrations $\mathcal{S}_1$ and $\mathcal{S}_2$.
		In addition, we suppose $\textbf{y}^{(1)}$ and $\textbf{y}^{(3)}$ are the optimal dispatching solutions under service orchestrations $\mathcal{S}_1 \cup \{(l_1, m_1)\}$ and $\mathcal{S}_2 \cup \{(l_1, m_1)\}$, respectively, that minimize $\sum_{i\in \mathcal{V}} \lambda_{p,l_1,i} y_{p,l_1,i,m_1}$. Then the objective function can be decomposed as:
		\begin{equation}
			\label{equation:9}
			\Omega(\mathcal{S}_1) = \sum_{(l,m) \in \mathcal{S}_1} \sum_{i\in \mathcal{V}} \lambda_{p,l,i}y_{p,l,i,m}^{(0)}  \tag {9}
		\end{equation}
		
		\begin{equation}
			\label{equation:10}
			\begin{array}{l}
				\Omega(\mathcal{S}_1 \cup \{(l_1, m_1)\}) =\\ \sum_{(l,m) \in \mathcal{S}_1} \sum_{i\in \mathcal{V}} \lambda_{p,l,i}y_{p,l,i,m}^{(1)}+\sum_{i\in V} \lambda_{p,l_1,i}y_{p,l_1,i,m_1}^{(1)} \tag {10}
			\end{array}
		\end{equation}
		
		\begin{equation}
			\label{equation:11}
			\Omega(\mathcal{S}_2) = \sum_{(l,m) \in \mathcal{S}_2} \sum_{i\in \mathcal{V}} \lambda_{p,l,i}y_{p,l,i,m}^{(2)} \tag {11}
		\end{equation}
		
		\begin{equation}
			\label{equation:12}
			\begin{array}{l}
				\Omega(\mathcal{S}_2 \cup \{(l_1, m_1)\}) =\\ \sum_{(l,m) \in \mathcal{S}_2} \sum_{i\in \mathcal{V}} \lambda_{p,l,i}y_{p,l,i,m}^{(3)}+\sum_{i\in V} \lambda_{p,l_1,i}y_{p,l_1,i,m_1}^{(3)} \tag {12}
			\end{array}
		\end{equation}
		
		Due to this decomposition, we have
		\begin{equation}
			\label{equation:13}
			\begin{array}{l}
				\Omega(\mathcal{S}_1 \cup \{(l_1, m_1)\}) - \Omega(\mathcal{S}_1)= \\ \sum_{(l,m) \in \mathcal{S}_1} \sum_{i\in \mathcal{V}} \lambda_{p,l,i}(y_{p,l,i,m}^{(1)}-y_{p,l,i,m}^{(0)})\\+\sum_{i\in V} \lambda_{p,l_1,i}y_{p,l_1,i,m_1}^{(1)} \tag {13}
			\end{array}
		\end{equation}
		
		\begin{equation}
			\label{equation:14}
			\begin{array}{l}
				\Omega(\mathcal{S}_2 \cup \{(l_1, m_1)\}) - \Omega(\mathcal{S}_2)= \\ \sum_{(l,m) \in \mathcal{S}_2} \sum_{i\in \mathcal{V}} \lambda_{p,l,i}(y_{p,l,i,m}^{(3)}-y_{p,l,i,m}^{(2)})\\+\sum_{i\in V} \lambda_{p,l_1,i}y_{p,l_1,i,m_1}^{(3)} \tag {14}
			\end{array}
		\end{equation}
	
		By in Lemma~\ref{lemma1}, there is no contention of computation resources between replicas, and hence replicas in $\mathcal{S}_1$ can still process requests dispatched to them under $\textbf{y}^{(0)}$. Thus, the first term in (\ref{equation:13}) is zero. Similarly, the first term in (\ref{equation:14}) is also zero. As there is no computation resource contention between replicas, requests that used to be served by replicas in $\mathcal{S}_1$ under service orchestration $\mathcal{S}_1 \cup \{(l_1, m_1)\}$ can still be served there after adding replicas in $\mathcal{S}_2 \setminus \mathcal{S}_1$, but these added replicas may offload some requests that used to be served by the replica $(l_1, m_1)$. Therefore, $\sum_{i\in V}\lambda_{p,l_1,i}y_{p,l_1,i,m_1}^{(1)}\geq\sum_{i\in V}\lambda_{p,l_1,i}y_{p,l_1,i,m_1}^{(3)}$ . This proves Eq. (8) and the sub-modularity of objective function in (\ref{equation:7}).
	\end{proof}
	The constraints of Eq. (7) also have a desirable property.~\cite{gupta2010}
	
	\begin{definition}
		\textit{Let $X$ be a universe of elements. Consider a collection $\mathcal{I} \subseteq 2^X$ of subsets of $X$. $(X, \mathcal{I})$ is called an independence system if: (a) $\emptyset \in \mathcal{I}$, and (b) if $Z \in \mathcal{I}$ and $Y \subseteq Z$, then $Y \in \mathcal{I}$ as well. The subsets in $\mathcal{I}$ are called independent; for any set $S$ of elements, an inclusion-wise maximal subset $T$ of $S$ that is in $\mathcal{I}$ is called a basis of $S$.}
	\end{definition}
	
	\begin{definition}\textit{Given an independence system $(X, \mathcal{I})$ and a subset $S \subseteq X$, the rank $r(S)$ is defined as the cardinality of the largest basis of $S$, and the lower rank $ \rho (S)$ is the cardinality of the smallest basis of $S$. The independence system is called a $p$-independence system (or a $p$-system) if max$_{S \subseteq X}$ $\frac{r(S)}{\rho(S)} \leq p$.}
	\end{definition}
	
	\begin{lemma}
		\label{lemma2}
		\textit{The constraints (7b) and (7c) form a $p$-independence system for} $p=\lceil \frac{\textrm{max} \ r_{p,l}}{\textrm{min}_{l: r_{p,l} \textgreater 0} \ r_{p,l}} \rceil$.
	\end{lemma}
	
	\begin{proof}
		According to definition 1, $(\mathcal{L}_p \times \mathcal{M}_p, \mathcal{I})$, where $\mathcal{I} \subseteq 2^{\mathcal{L}_p \times \mathcal{M}_p}$ is a set of all feasible solutions to Eq. (7), is an independent system, as the subset of any feasible service orchestration remains feasible. We suppose any $\mathcal{S} \subseteq \mathcal{L}_p \times \mathcal{M}_p$ and any two maximal feasible service orchestrations $\mathcal{S}_1, \mathcal{S}_2$, where $\mathcal{S}_1 \subseteq \mathcal{S}_2$. To add a pair $(l, m) \in \mathcal{S}_2 \setminus \mathcal{S}_1$ to $\mathcal{S}_1$, we need to take out a set $\mathcal{S}^{'}$ of pairs from $\mathcal{S}_1$ , such that $(\mathcal{S}_1 \setminus \mathcal{S}^{'}) \cup \{(l, m)\}$ remains a feasible service orchestration. The set $\mathcal{S}^{'}$ contains at most $\lceil \frac{\textrm{max} \ r_{p,l}}{\textrm{min}_{l: r_{p,l} \textgreater 0} \ r_{p,l}} \rceil$ pairs corresponding to removing service replicas from resource cell $m$ to satisfy constraint (7b). In the process of modifying $\mathcal{S}_1$ into $\mathcal{S}_2$ by repeating these swaps, we can reduce the number of orchestrated service replicas by at most $p$-fold. Consequently, the constraints in Eq. (7) form a $p$-independent system.
	\end{proof}
	
	Combining Lemmas~\ref{lemma1} and~\ref{lemma2} gives the following result.
	
	\begin{theorem}
		\label{theorem1}
		\textit{Under the conditions in Lemma ~\ref{lemma1},} Solve JSORD Based on Submodular Function Maximization \textit{(Algorithm 1) yields a} 1/(1 + $p$)-\textit{approximation for Eq. (1), where} $p=\lceil \frac{\textrm{max} \ r_{p,l}}{\textrm{min}_{l: r_{p,l} \textgreater 0} \ r_{p,l}} \rceil$.
	\end{theorem}
	
	\begin{proof}
		A greedy algorithm for the monotonic sub-modular case that works for p-independence systems is presented in \cite{fisher1978}, where the approximation ratio is 1/($p$ +1).
	\end{proof}

\begin{figure*}[t]%
	\centering
	\includegraphics[width=0.8\linewidth]{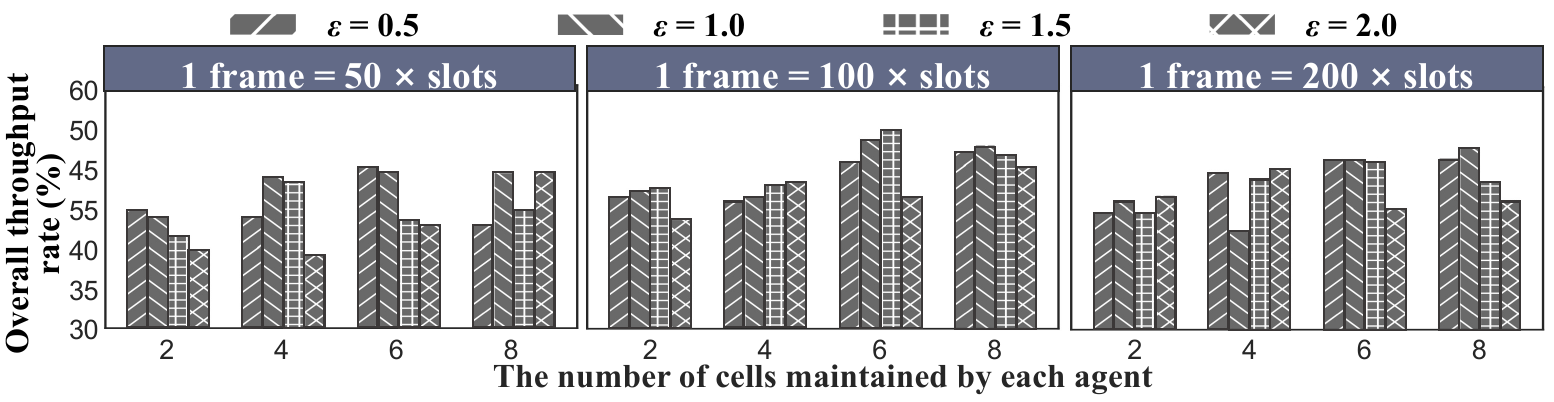}
	\caption{\textit{EdgeMatrix}'s training performance under various parameters.}
	\setlength{\abovecaptionskip}{-0.5cm} 
	\label{fig:ParametersIdentify}
	\vspace{-1.2em}
\end{figure*}

\section{Implementation}
\label{sec:System Design and Implementation}

\subsection{Multi-tier System Setup}
\label{subsec:End-cloud-edge Computing System Setup}

\textbf{\textit{Physical resources}}. In the simulation experiment, we set up $10$ edge nodes and $6$ SLA Allocation/Retention Priority (ARP) services\cite{3GPPwhitepaper} in the edge cluster of a region. \vs{In addition, a fog node is set up for management, which is not used for deploying services but only for aggregating global observations and deploying optimization algorithms. At runtime, each edge node periodically aggregates its own observations to the fog node according to the two-time-scale framework, and then the agent in the fog node will make decisions and send them to the corresponding edge node for execution.} Each property of the edge node $i \in \mathcal{V}$ is set to $W_{i}=\left[2,4\right]$ vCPUs, $R_{i}=\left[100, 200\right]$ GB and $B_{i}=\{125,12.5\}$ Mbps. We assume that the computing capacity $W_{\rm{cloud}}$ and memory capacity $R_{\rm{cloud}}$ of cloud center are always sufficient, and the connection between the cloud center and the edge cluster is reliable, so the transmission delay from the edge node to the cloud center is set to a constant $L_{\rm{edge}}^{\rm{cloud0}}=10 ms$.

\textbf{\textit{Logically resources}}. By default, at each frame $\tau$, the number of resource cells that each agent can maintain is $\sum_{p\in \mathcal{P}}m_{p,i} \geq 1$. The requirements of each resource cell are predicted by the continuous action $\hat{a}_{i,\tau}=\{W_{m_{i,\tau}}, R_{m_{i,\tau}}\}$. These two numbers are floats between $[0,1]$, and their true resource requirements are multiplied by the scalar in the system, i.e., $(2\ \rm{vCPUs}, 500\ \rm{MB})$. These resource cells are classified to the corresponding resource channels according to the cell characteristics $\Phi$. There are $[2,4]$ kinds of service on each channel, and the number of resource channels is $6$.

\textbf{\textit{Service and request}}. \vs{Our system's data value range and request distribution are based on the Alibaba Cluster Trace\cite{Alibabacluster} and the Google Cluster Trace~\cite{google} to ensure \textit{EdgeMatrix} has effective performance in the real environment.} However, since the dataset does not reflect the delay characteristics, we refer to ETSI's white paper\cite{3GPPwhitepaper} to determine the delay data range of requests.

\subsection{Training Settings}
\label{subsec:Training Settings}

We implement Algorithm 2 using python 3.6 and set up the details as follows. Each NMAC agent consists of a critic network and an actor policy network, and the network is trained with a fixed learning rate of $\eta=0.01$ and the reward discount factor of $\gamma=0.95$. 
The critic network consists of a three-layer fully connected neural network and has 64 neurons per layer, where the activation function is relu in the first two hidden layers and no activation function in the output layer. 
Similarly, the actor policy network is also consists of a three-layer fully connected neural network and has 64 neurons per layer, with the activation function is relu for the first two layers and the sigmoid for the output layer to ensure that the output is in the valid range. In addition, the linear program solver in JSORD uses linprog function of the SciPy library.

\begin{figure}[t]%
	\centering
	\includegraphics[width=1.0\linewidth]{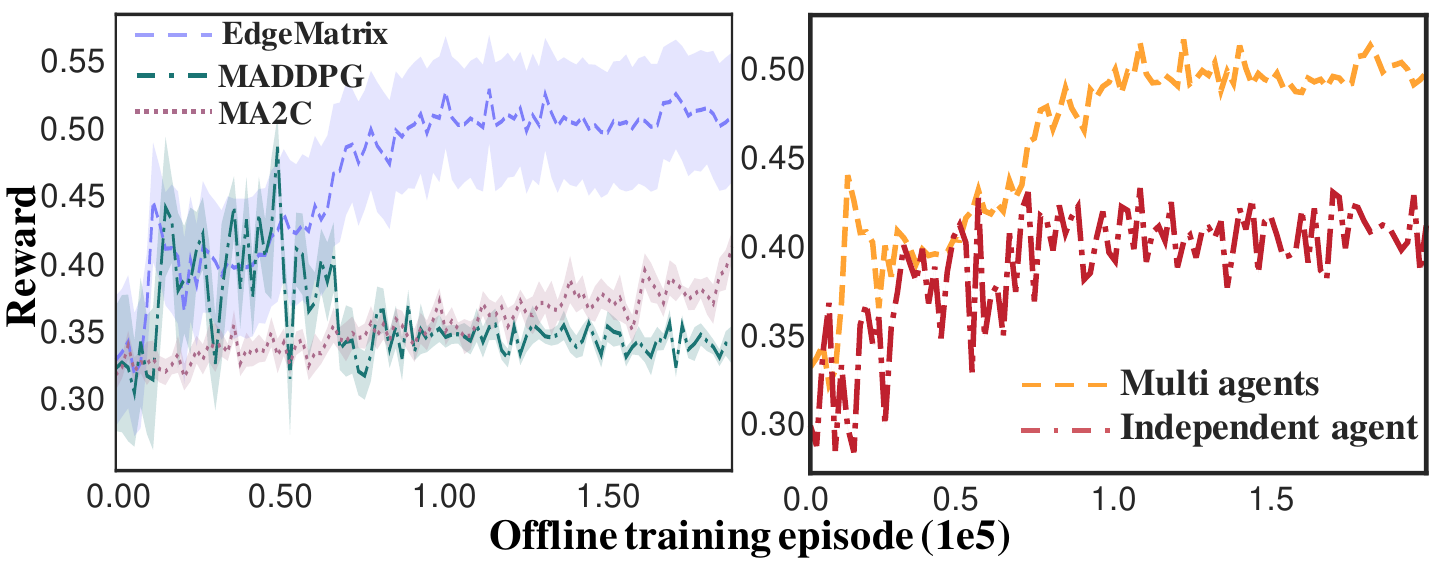}
	\setlength{\abovecaptionskip}{-0.5cm} 
	\caption{Learning ability of \textit{EdgeMatrix} compared with (a) the two baselines, MADDPG and MA2C and (b) independent PPO and independent DQN.}
	\label{fig:LearningAbility}
	\vspace{-0.6em}
\end{figure}

\section{Trace-driven Performance Evaluation}
\label{sec:Performance Evaluation}

\subsection{Setting of Key Parameters}
\label{subsec:Setting of Key Parameters}

As shown in Fig. \ref{fig:ParametersIdentify}, we first need to determine several important parameters for the training. The frequency of service orchestration has a significant impact on the training performance of \textit{EdgeMatrix}.
(\textit{$\romannumeral1$}) For \textit{frequent orchestration}, the cell must reload replicas of the service at each service orchestration, resulting in high service deployment/deletion costs; (\textit{$\romannumeral2$}) For \textit{infrequent orchestration}, since the requests of the networked system are constantly changing, it will be difficult for the service to adapt to the system dynamics. Based on Fig. \ref{fig:ParametersIdentify}, we define 100 slots per frame. In addition, we found that the system performance is better as the number of cells maintained per node increases, but when the number is too high, it is not beneficial to the system performance. Therefore, we define the number of cells maintained by each node is 6. As one of the basic features of SLA prioritization, the edge resource proportion determines the latency guarantee of the cell. However, the large weight of the edge resource proportion can result in negligible impact on core network resources, so we set it as $\mathcal{\varepsilon}=1.5$ based on Fig. \ref{fig:ParametersIdentify}.

\begin{figure*}[t]%
	\centering
	\includegraphics[width=0.85\linewidth]{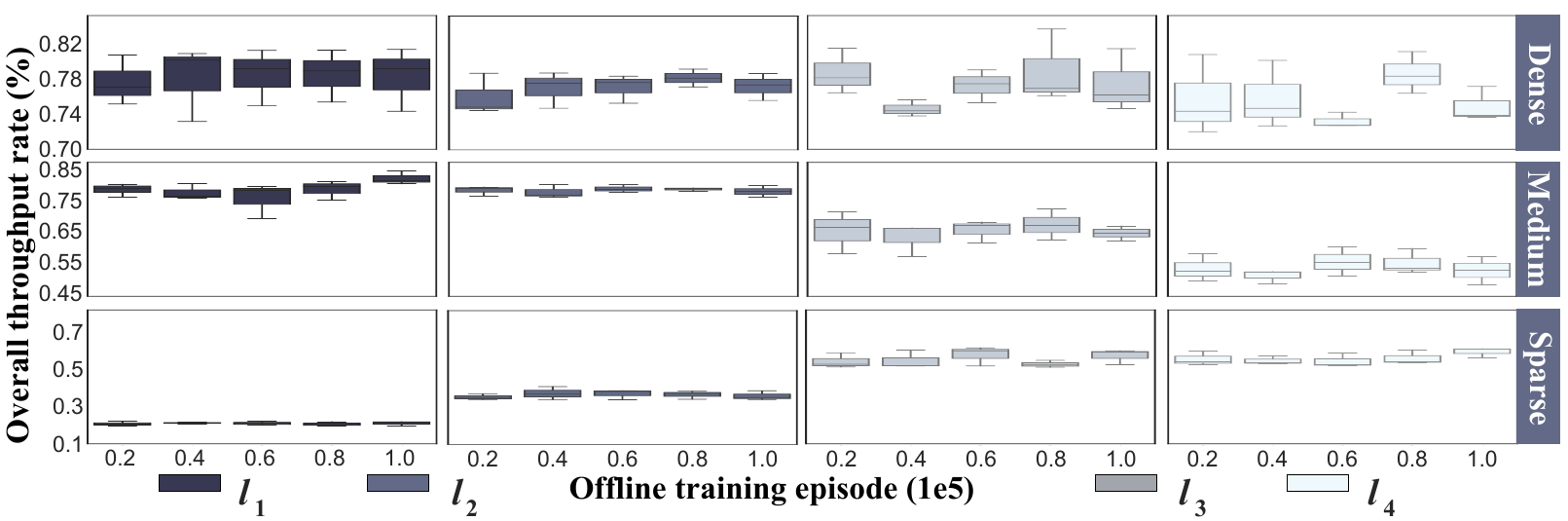}
	\caption{The impact of distribution heterogeneity on the performance of \textit{EdgeMatrix} in the case of three edge cluster densities: (a) dense (left), (b) medium (middle), (c) and sparse (right).}
	\label{fig:distance}
	\vspace{-1.8em}
\end{figure*}

\subsection{Learning Ability of \textit{EdgeMatrix}}

To verify the learning performance of \textit{EdgeMatrix}, we first conducted a comparison of its training performance with other baselines to demonstrate that \textit{EdgeMatrix} has the feasibility of convergence and effective learning ability. Then we compare the performance of the algorithm \textit{EdgeMatrix}, independent DQN\cite{mnih2015}, and independent PPO\cite{schulman2017} to demonstrate the effectiveness of \textit{EdgeMatrix} among 10 edge nodes.

As shown in Fig. \ref{fig:LearningAbility}(a), the three algorithms at the first 100 training episodes are in the exploration stage with random policy, i.e., they do not learn policies and only collect training data, and their rewards are at the same level. 
The rewards of \textit{EdgeMatrix} and MADDPG started to rise sharply after the 100th episode and remained flat after that. This is because the algorithms are able to learn some lessons from the data collected in the first 100 episodes. However, the learning performance is flattening along with the number of training times increases and both converge after the 1000th episode. In particular, the performance of \textit{EdgeMatrix} is improved by about 60\%, but MADDPG learns almost nothing from the experience in the data. Moreover, the performance improvement of MA2C is small, which means that the algorithm is not adaptable to scenarios.  
Fig. \ref{fig:LearningAbility}(b) demonstrates that a simple implementation of independent agent in the multi-agent environment is not excellent due to the non-stationary problem.

	\subsection{Impact of Geographical Distribution}
	Since one of the core differences between edge and cloud paradigms is the geographical distribution, we further explore the impact of it on \textit{EdgeMatrix}.  
	Firstly, we divide three scenarios including dense, medium and sparse based on the density of edge nodes, i.e., the average number of nodes per unit area.  
	Afterward, we compare the \textit{EdgeMatrix} performance in each edge node density scenario with different distributional heterogeneity, i.e., the clustering degree of edge nodes in terms of geographical distribution, where the distributional heterogeneity gradually increases from $l_1$ to $l_5$.
	
	As shown in  Fig. \ref{fig:distance}(a), in a dense edge node distribution scenario \textit{EdgeMatrix} has sufficient edge node resources to rely on, so it always performs well, and geographical distribution heterogeneity has little impact on performance, which is the ideal state to deploy \textit{EdgeMatrix}.
	
	In most cases, the edge node density can not reach a high enough degree. Fig. \ref{fig:distance}(b) shows that when the density of edge nodes is medium, the more evenly distributed \textit{EdgeMatrix} has better performance. This is because the geographical location of task nodes is nearly evenly distributed, so the smaller distributional heterogeneity of edge nodes can help reduce the average distance between task nodes and the nearest edge nodes, so as to a low latency response. 

	In addition, we also consider the scenario where the number of edge nodes is sparse. As shown in Fig. \ref{fig:distance}(c), when edge node resources are scarce, the more significant the heterogeneity of the node distribution is, the more beneficial it is to improve the performance of \textit{EdgeMatrix}. The reason is that the resources of a single edge node are too low to meet the needs of the service. However, higher heterogeneity of geographical distribution leads to more clustering of edge nodes in geographical distribution. Compared with a single edge node, clustered multiple edge nodes have higher resource advantages, thus ensuring resource supply in more cases.  

\subsection{Practicability of \textit{EdgeMatrix}}

Fig. \ref{fig:Practicability} validates two sub-objectives of \textit{EdgeMatrix}: (\textit{$\romannumeral1$}) maximizes the overall throughput;  (\textit{$\romannumeral2$}) reducing the SLA violation for various services. We compare the overall throughput rates of different algorithms in Fig. \ref{fig:Practicability}(b) under the request distribution from the Alibaba Cluster Trace~\cite{Alibabacluster} (as shown in Fig. \ref{fig:Practicability}(a)), showing the percentage of requests served by each channel to all requests in Fig. \ref{fig:Practicability}(c).
As shown in Fig. \ref{fig:Practicability}(c), \textit{EdgeMatrix} performs 36.7\% better than the nearest baseline with the same request distribution. In all six channels (1-6), the smaller the value of Channel\_Id, the higher the SLA priority that can be guaranteed for that channel. 
Among them, the channels with Channel\_Id (1-3) have horizontal cells, which means that the orchestrated services have high SLA priority and low transmission delay, and the channels with channel Channel\_Id (4-6) have vertical cells. 

In addition, with the weight of edge resources proportion set during our training process (i.e., $\mathcal{\varepsilon}=1.5$), the number of requests served by the horizontal channel accounted for 73.7\% of the total number of requests. The number of requests served by the vertical channel is accounted for 26.3\%. Note that the throughput rate of the different services can be adjusted by $\mathcal{\varepsilon}$.
\vs{In addition, to enhance the persuasiveness of the experimental results, we used the Google Cluster Trace~\cite{google} as a supplement, i.e., the dataset-driven request requirements as shown in Fig. \ref{fig:google} (a). Further, it can be found that \textit{EdgeMatrix} is 31.6\% higher than the closest baseline as shown in Fig. \ref{fig:google} (b). This demonstrates that \textit{EdgeMatrix} is well adapted to the dynamics of the system and can adjust its policy based on real-time request demand patterns.}
\begin{figure}[t]%
	\centering
	\includegraphics[width=1.0\linewidth]{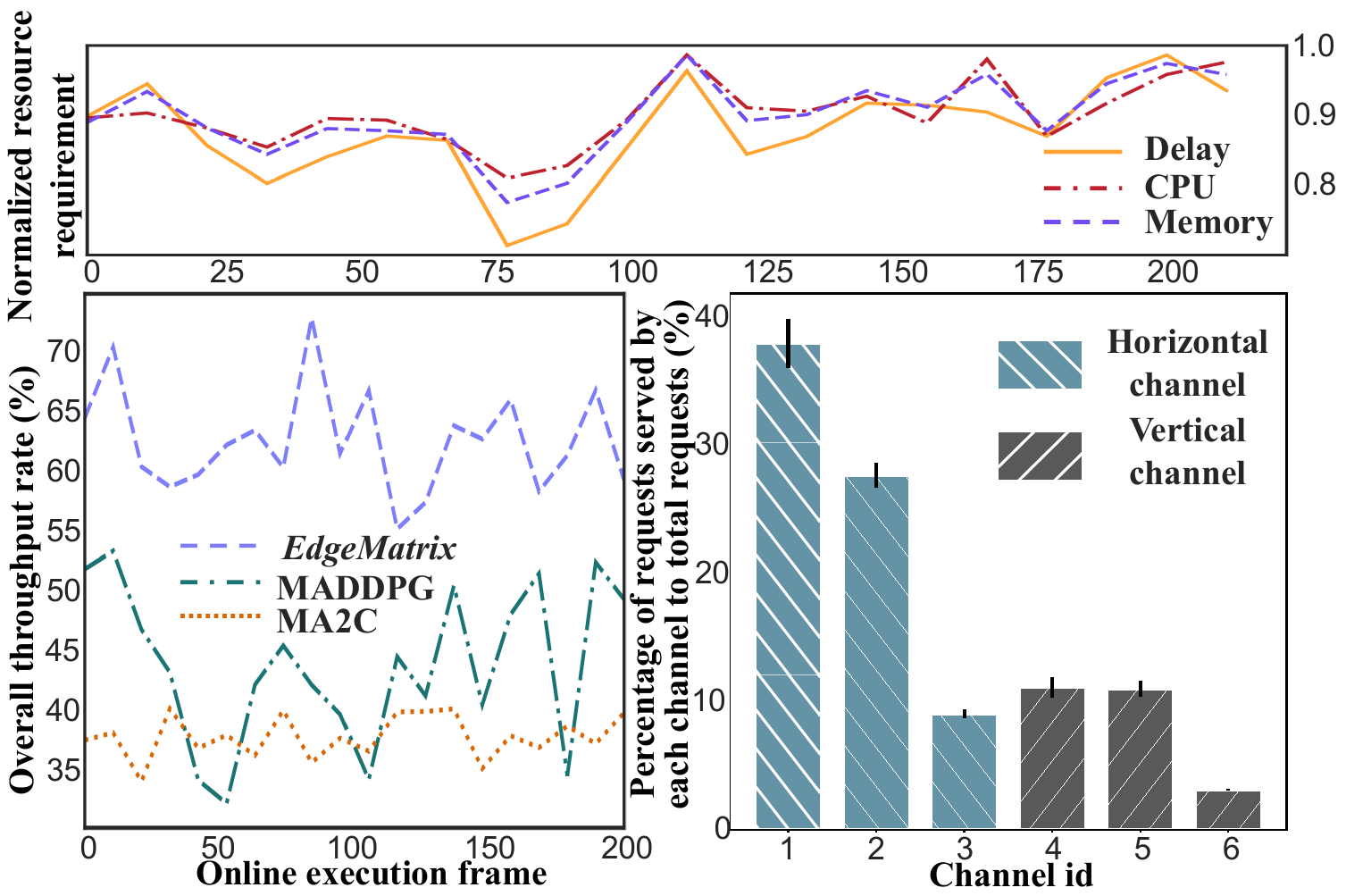}
	\setlength{\abovecaptionskip}{-0.5cm} 
	\caption{(a) Under the stochastic request arrivals (top), (b) the overall throughput rate of \textit{EdgeMatrix} against baselines (left buttom), and (c) the percentage of requests served by each channel to total requests (right buttom).}
	\label{fig:Practicability}
	\vspace{-1.3em}
\end{figure}

\subsection{Performance in Complex Environment}

To verify the adaptability of \textit{EdgeMatrix} in the multi-tier system, we evaluated the performance of \textit{EdgeMatrix} under three inherent challenges: multi-resource heterogeneity, resource competition, and networked system dynamics.

Fig. \ref{fig:Performance}(a) shows the performance of \textit{EdgeMatrix} under the resource heterogeneity of edge nodes. While keeping the total amount of each resource unchanged, we vary the variance of resources between edge nodes and classify the resource heterogeneity into five levels (1-5) according to the variance, where larger values imply higher heterogeneity. 
\vs{We find that \textit{EdgeMatrix} performs better when the heterogeneity levels of computational and memory resources are the same. It is due to the correlation between requests for computational resources and memory resource requirements, i.e., the request with high memory resource requirement will have high computing resource requirement with high probability. In addition, the effect of heterogeneity on \textit{EdgeMatrix} is small, with only a 3.9\% decrease in the strongest case of edge node resource heterogeneity compared to the weakest case.}

\textbf{\textit{EdgeMatrix under resource competition}}.
Resource requirements have a significant impact on throughput. When the load level of computational and memory resources is higher, the competition for such resources becomes more intense. Fig. \ref{fig:Performance}(b) shows that \textit{EdgeMatrix} is able to adapt to dynamic changes in the degree of resource competition, and its ability to adjust to memory resource competition is better than that of computational resources. 
\vs{This is because the extra computing resources can help to improve the efficiency of request processing, while the elasticity of memory resources is weak, i.e., the extra memory resources have little help to the request processing. 
In addition, it also shows that \textit{EdgeMatrix} benefits from the isolation capability of channels and the online learning capability of MADRL to sense the load changes of various resources in the environment and adjust the policy in time, thus maintaining efficient resource allocation capability in the dynamic resource competition.
}

\textbf{\textit{EdgeMatrix under networked system dynamics}}.
The size of bandwidth resources has a direct impact on the stability of the network system. The heterogeneous level of bandwidth is set in the same way as before. The higher the level (1-5) of the average network bandwidth in the system, the higher the average bandwidth will be. 
\vs{Fig. \ref{fig:Performance}(c) shows that the throughput rate of the service increases as the bandwidth increases and decreases as the heterogeneity of the bandwidth resources increases. Also, the larger the average bandwidth resource, the less affected the heterogeneity change. This indicates that the high amount of resources can compensate for the negative impact caused by resource heterogeneity. In addition, it also verifies that the resource reorganization of \textit{EdgeMatrix} plays a positive role in effectively controlling the edge nodes regardless of whether their resources are too large or too small.}

\begin{figure}[t]%
	\centering
	\includegraphics[width=1.0\linewidth]{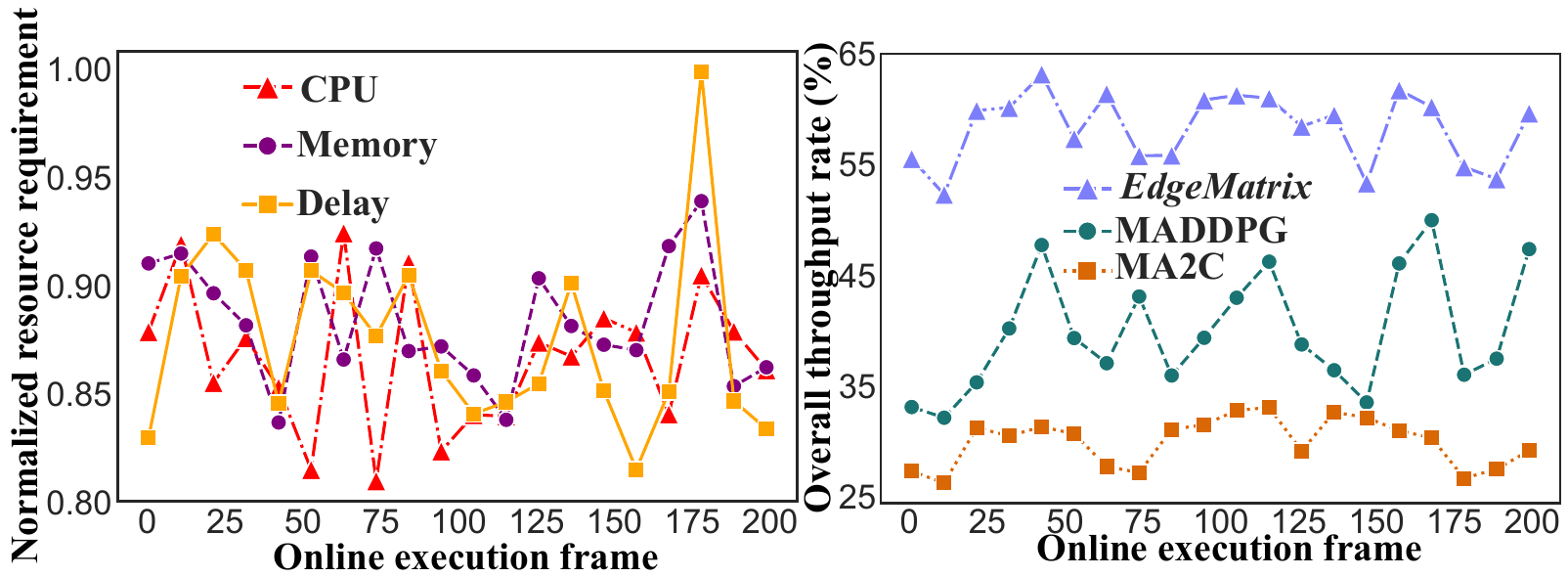}
	\setlength{\abovecaptionskip}{-0.5cm} 
	\caption{\vs{(a) Under the request requirements (right), (b) the overall throughput rate of \textit{EdgeMatrix} against baselines (left).}}
	\label{fig:google}
	\vspace{-0.6em}
\end{figure}

\begin{figure*}[t]%
	\centering
	\includegraphics[width=1.0\linewidth]{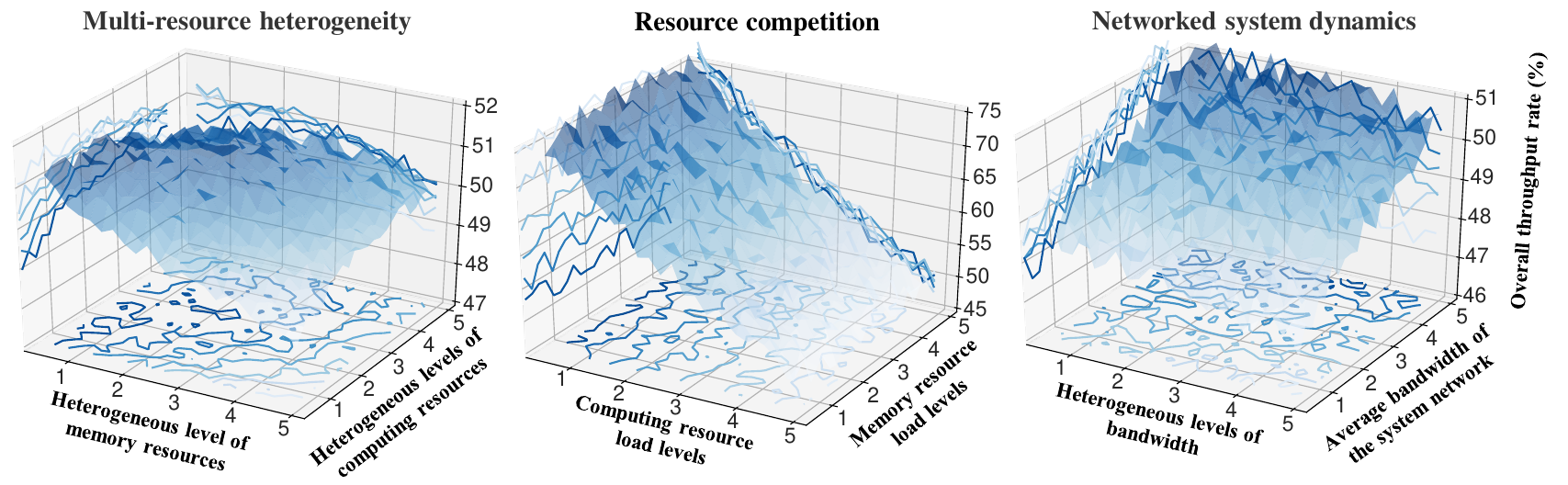}
	\caption{Practicability of \textit{EdgeMatrix} under (a) hetergeneous resources contains cpu and memory (left), (b) the resource competition (central), and (c) the dynamic networked environment (right).}
	\label{fig:Performance}
\end{figure*}

\begin{figure}[t]%
	\centering
	\includegraphics[width=1.0\linewidth]{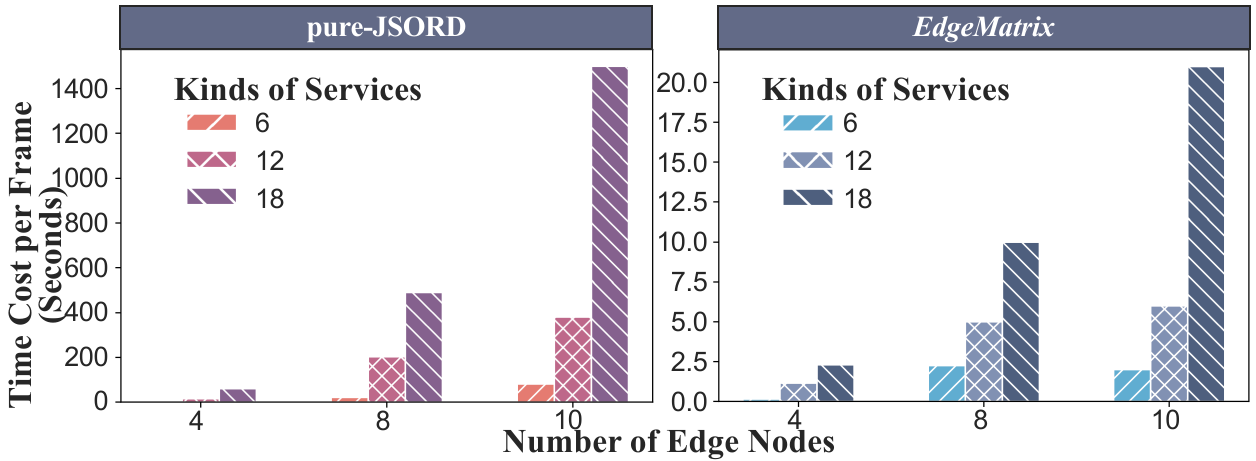}
	\caption{\vs{Runtime comparison between \textit{EdgeMatrix} and pure-JSORD.}}
	\label{fig:Time Complexity}
	\vspace{-0.6em}
\end{figure}

\subsection{Runtime Cost Reduction}

One of the most important contributions of \textit{EdgeMatrix} is to significantly reduce the runtime of service orchestration and request dispatch. As shown in Fig. \ref{fig:Time Complexity}, we compare the decision time cost by \textit{EdgeMatrix} and pure-JSORD to perform service orchestration and request dispatch for each frame under different numbers of nodes and service types when the number of channels is 6.
We found that the time required by pure-JSORD and \textit{EdgeMatrix} increases with the number of nodes and service type, but the time cost required by \textit{EdgeMatrix} is much lower than that of pure-JSORD. We observe that the runtime of pure-JSORD is 13 to 71 times higher than \textit{EdgeMatrix} for a small range of parameter values.
The reason is that the traditional method considers all services and requests within the global nodes, unlike \textit{EdgeMatrix} which (\textit{$\romannumeral1$}) divides the SLA priority levels of the services and orchestrates them with the corresponding SLA priority on each channel; (\textit{$\romannumeral2$}) dispatches requests to cells on each channel rather than the global edge nodes, and only dispatch requests with a corresponding SLA priority. These features allow \textit{EdgeMatrix} to perform service orchestration and request dispatch in parallel between channels and significantly reduce the magnitude of parameters.

\section{Related Work}
\label{sec:Related Work}

\textbf{\textit{Resource Customization}}. \vs{Our work is inspired by the design concept of network slicing~\cite{wu2022survey}, i.e., using SDN and NFV technologies to map resources in physical infrastructure to dedicated virtual resources. Further, customized services and resource isolation are provided to efficiently utilize the limited resources in the network system, such as RANs\cite{das2022optimal,gholivand2021cloud, elhattab2021edge} and Core Network\cite{bega2019,martin2020}.} However, some existing research considered a separate MEC host \cite{zhang2019, chantre2020} or Service Chain Functions (SCFs) \cite{pham2021optimized,luu2020} in the edge node as a slice for the multi-tier system. However, they do not fully consider the multi-resource heterogeneity in edge environments.

\textbf{\textit{Joint Service Orchestration and Request Dispatch}}. At present, many researchers have been concerned about how to use resources in multi-tier system reasonably and efficiently through service orchestration and request dispatch. 
\vs{In~\cite{lv2022microservice}, the authors consider the service orchestration problem in edge computing scenarios. To minimize the communication overhead, the authors design a DQN-based algorithm to deploy services and a heuristic-based elastic scaling algorithm to adjust the number of containers. In~\cite{jovsilo2020computation}, the authors consider the request dispatching problem in edge computing scenarios and design a game-theoretic-based request dispatching algorithm to minimize cost. However, the algorithm design in both~\cite{lv2022microservice} and~\cite{jovsilo2020computation} ignores the interrelationship between requests and services to design independent solutions.}
In addition, the authors of \cite{wang2022joint} take a more comprehensive perspective to realize the joint optimization design of task scheduling, service caching in the multi-tier system, but the heterogeneity of physical resources and service SLA guarantees are not sufficiently considered.

\section{Conclusion}
\label{sec:Conclusion}

In this paper, we propose \textit{EdgeMatrix} to implement service SLA prioritization guarantees for multi-tier systems under three challenges: multi-resource heterogeneity, resource competition, and networked system dynamics. First, \textit{EdgeMatrix} reduces the complexity of optimizing multi-tier systems by solving NMAC to re-customize physical resources into relatively independent resource cells and resource channels. After that, \textit{EdgeMatrix} further reduces the time required for algorithmic decision-making by executing JSORD in each resource channel in parallel. As a result, \textit{EdgeMatrix} has good performance in large-scale multilayer systems, i.e., \textit{EdgeMatrix} is 36.7\% better than the closest baseline.